\documentclass[pdflatex,sn-mathphys]{sn-jnl}

\theoremstyle{thmstyleone}%
\newtheorem{theorem}{Theorem}
\newtheorem{proposition}[theorem]{Proposition}%
\newtheorem{lemma}[theorem]{Lemma}
\newtheorem{corollary}[theorem]{Corollary}

\theoremstyle{thmstyletwo}%
\newtheorem{remark}{Remark}%

\theoremstyle{thmstylethree}%
\newtheorem{definition}{Definition}%

\usepackage[T1]{fontenc}

\usepackage[utf8]{inputenc} 

\usepackage{amssymb,mathtools}
\usepackage{graphicx}
\usepackage{amsthm}
\usepackage{amssymb, latexsym, bm, mathtools, multirow,enumitem}

\usepackage[cmintegrals]{newtxmath}
\usepackage{color}

\theoremstyle{plain}

\newcommand{\cH}{\mathcal{H}}

\newcommand{\cO}{\mathcal{O}}

\newcommand{\ket}[1]{\vert #1 \rangle}
\newcommand{\bra}[1]{\langle#1\vert}
\newcommand{\ketbra}[2]{\ket{#1}\!\bra{#2}}
\newcommand{\braket}[2]{\langle#1\vert#2\rangle}

\DeclareMathOperator{\tr}{Tr}

\newcommand{\id}{\text{id}}

\def\cH{\mathcal{H}}

\newcommand{\dm}[1]{\vert#1\rangle\langle#1\vert}
\newcommand{\bal}{\begin{equation}\begin{aligned}}
\newcommand{\eal}{\end{aligned}\end{equation}}

\newcommand{\mB}{\mathcal{B}}
\newcommand{\mH}{\mathcal{H}}

\newcommand{\base}[1]{#1}

\def\QED{\mbox{\rule[0pt]{1.5ex}{1.5ex}}}
\newenvironment{proofof}[1]{\vspace*{5mm} \par \noindent
\quad{\it Proof of #1: \hspace{2mm}}}{\hspace*{\fill}~\QED}

\def\Label{\label}

\begin{document}

\title[When quantum memory is useful]{When quantum memory is useful for dense coding}

\author*[1,2]{\fnm{Ryuji} \sur{Takagi}}\email{ryujitakagi.pat@gmail.com}

\author*[3,4,5]{\fnm{Masahito} \sur{Hayashi}}\email{hmasahito@cuhk.edu.cn}

\affil[1]{\orgdiv{Department of Basic Science, Graduate School of Arts and Sciences}, \orgname{The University of Tokyo}, \orgaddress{3-8-1 Komaba, Meguro-ku, \city{Tokyo}, \postcode{153-8902}, \country{Japan}}}
\affil[2]{\orgdiv{Nanyang Quantum Hub, School of Physical and Mathematical Sciences}, \orgname{Nanyang Technological University}, \orgaddress{\postcode{637371}, \country{Singapore}}}

\affil[3]{\orgdiv{School of Data Science, The Chinese University of Hong Kong, Shenzhen}, \orgaddress{ \city{Shenzhen}, \postcode{518172}, \country{China}}}

\affil[4]{\orgname{International Quantum
Academy (SIQA)}, \city{Shenzhen}, \postcode{518048}, \country{China}}

\affil[5]{\orgdiv{Graduate School of Mathematics}, \orgname{Nagoya University}, \orgaddress{\city{Nagoya}, \postcode{464-8602}, \country{Japan}}}


\abstract{We discuss dense coding with $n$ copies of a specific preshared state between the sender and the receiver when the encoding operation is limited to the application of group representation. Typically, to act on multiple local copies of these preshared states, the receiver needs quantum memory, because in general the multiple copies will be generated sequentially. Depending on available encoding unitary operations, we investigate what preshared state offers an advantage of using quantum memory on the receiver's side. 
}




\maketitle


\section{Introduction}

Quantum dense coding is a key subroutine in many quantum information processing tasks.   
Dense coding utilizes preshared entanglement among the sender and receiver to communicate more classical bits than the number of quantum bits sent from the sender to receiver \cite{bennett1992communication,hiroshima2001optimal,bowen2001classical,horodecki2001classical,winter2002scalable,bruss2004distributed,beran2008nonoptimality,horodecki2012quantum,datta2015second,das2019quantum,laurenza2019dense,wakakuwa2020superdense,hayashiWang}.
It is well known that one preshared maximally entangled qubit together with one bit of quantum communication admits communicating two bits of classical information. 
Although it is evident that preshared entanglement is a necessary resource to enable dense coding, \emph{how} entanglement helps the enhancement is still unclear --- the presence of entanglement in a preshared state does not necessarily guarantee to enhance classical communication for a given encoding strategy. 
It is therefore important to characterize when entanglement indeed becomes helpful for dense coding. 

This question can be approached by observing that holding entanglement during communication requires the receiver to store the preshared state in a quantum memory. 
If the receiver is not in possession of quantum memory but instead makes a measurement right after the preshared state is provided, entanglement is broken before the communication is completed.  
Therefore, characterizing when quantum memory enhances the classical communication rate gives insights into the role genuinely played by quantum entanglement.  
This question also has practical relevance, as preparing a long-term quantum memory may come with much technical difficulty~\cite{Zhao2009long-lived,Khabat2016quantum}. 

As the simplest setting, we restrict our encoding operations to a certain subgroup
under the presence of preshared states \cite{hayashiWang} ---
we assume that the encoding operation is given as a (projective) unitary representation $U$ of a group $G$ on $\cH_A$.
When the preshared entangled state is written as $\rho_{AB}$
and our coding operation is restricted to these unitaries,
our channel can be written as the classical-quantum (cq) channel
$g \mapsto U_g \rho_{AB} U_g^\dagger$.
Since this cq channel has a symmetric property for the group $G$,
we say that it is a {\it cq-symmetric} channel\footnote{The class of cq-symmetric channels 
has been studied well in the literature~\cite{korzekwa2019encoding}.
It is a quantum generalization of
a regular channel \cite{delsarte1982algebraic}, which was 
introduced as a useful class of channels in classical information theory.
This class of classical channels has various names, such as generalized additive \cite[Section V]{hayashi2011exponential} and 
conditional additive \cite[Section 4]{hayashi2011exponential} channels etc.
The class of additive channels is contained in the regular channels as a subclass.
The reference \cite[Section VII-A-2]{hayashi2015quantum} studied its quantum extension with an additive group
and discussed
the capacity and the wire-tap capacity with semantic security.
Since this class has a good symmetric property, it is known that
algebraic codes achieve the capacity \cite{delsarte1982algebraic,dobrushin1963asymptotic,elias1955coding,hayashi2020finite,hayashi2015quantum}.}. 
Recently, Ref. \cite{korzekwa2019encoding} studied such a model in the context of the resource theory of asymmetry
without considering shared entanglement.
Ref. \cite{hayashiWang} applied 
cq-symmetric channel to the above setting under the presence of shared entanglement.

In Ref.~\cite{hayashiWang}, two scenarios are considered, the first where Bob has to wait for the quantum communication from Alice before he can perform a global measurement on all systems, and a second where Bob is allowed to apply a global measurement on his local part of the shared states before he receives Alice's communication. While the second system reduces memory requirements at Bob, Bob still requires quantum memory when the preshared states are generated one by one.

In this work, in order to model the lack of quantum memory, we consider a third case where Bob measures his local part of the entangled states one by one, and investigate when the capacities of these three cases coincide.
The practical meaning of the first and third cases
are compared in Sec.~\ref{sec:standard capacity}, Remark \ref{Rem1}.
In particular, we study the conditions for a preshared entangled state when the capacities for the second and third case differ under which quantum memory becomes helpful in dense coding, whose encoding is realized by the action of a given group representation. 
We also relate our setting to private dense coding with the existence of an eavesdropper, extending the private capacity characterized for the first model, i.e., the case when the receiver is allowed to make a collective measurement on the joint system~\cite{Wu-Long-Hayashi}.

We begin by studying the setting involving an arbitrary preshared state and an arbitrary group representation, where we provide necessary (and sufficient) conditions for a quantum memory to be useful. 
Using a necessary condition, we show an example of an entangled state, in which the capacities are equal and
quantum memory is useless in this sense. 
We then focus on pure preshared states and give an explicit structure of the preshared state such that the receiver does not need to hold a quantum memory to realize the optimal decoding.  
Imposing the multiplicity-free condition on the group representation, 
we establish a novel connection between dense coding and the resource theory of speakable coherence~\cite{streltsov} --- equivalent to the resource theory of asymmetry with U(1) group with a multiplicity-free representation --- which brings us a further characterization of such pure preshared states.

In fact, the resource theory of asymmetry is a topic to study physical resources for information processing and has been under active investigation~\cite{bartlett2007reference,gour2008resource,gour2009measuring,bartlett2009quantum,marvian2012symmetry,marvian2013theory,marvian2014asymmetry,marvian2014extending,marvian2014modes,wakakuwa2017symmetrizing,wakakuwa2020superdense}.
In the context of quantum hypothesis testing, the references \cite{hiai1991proper} and \cite[Theorem 2.9]{hayashi2017group}
considered the difference with respect to the relative entropy between 
a state and the averaged state under a certain group action, where the averaged state is the output
of the noisy channel composed of a random unitary application. 
This relative entropy shows the difficulty of distinguishing the noiseless state and the noisy averaged state under a certain group action.
Here, we discuss this type of discrimination as a modification of quantum illumination, which was originally introduced for enhancing the capability of detecting a target object with low reflectivity using entanglement~\cite{Lloyd2008enhanced,Tan2008quantum}.
In particular, we apply our conditions to characterize when the entanglement in a given input state enhances the performance in guessing whether the noiseless or the group twirling channel acted on the entanglement half. 
We find that in the case of asymmetric discrimination with an abelian group, the performance with a maximally entangled input cannot be enhanced by local quantum memory, which can be contrasted to the conventional quantum illumination, in which quantum memory is essential when the input is a two-mode squeezing vacuum state~\cite{dePalma}.

The remainder of this paper is organized as follows.
Section \ref{S2} formulates the two cases of interest for quantum dense coding.
Section \ref{S2-1} formulates the case with quantum memory allowing for the global measurement, and 
Section \ref{S2-3} formulates the case without quantum memory
after 
Section \ref{S2-2} prepares the notations for group representation.
Section \ref{S3} discusses when quantum memory is useless for a general preshared state
and presents several examples of 
such entangled states.
Section \ref{S4} investigates simpler conditions for pure states
and presents relevant examples.
This analysis reveals the relation between the usefulness of quantum memory and the genuinely incoherent operations (GIO), a class of operations introduced in the resource theory of speakable coherence~\cite{deVicente}.
Section \ref{S5B} applies obtained results to the case of the maximally entangled states.
Section \ref{S5} applies our results to 
a modified version of quantum illumination.
Section \ref{S7} 
provides discussion for our results and suggests future studies.


\section{Dense coding with a preshared state}\Label{S2}
\subsection{Dense coding with quantum memory}\Label{S2-1}

\subsubsection{Standard capacity}

We formulate dense coding with a preshared state as in Ref.~\cite{hayashiWang},
which is based on a cq-symmetric channel \cite{korzekwa2019encoding}.
Assume that Alice and Bob share $n$ copies of the quantum state $\rho_{AB}$
on the system ${\cal H}:=\mathcal{H}_A \otimes \mathcal{H}_B $,
where $\mathcal{H}_B $ is a $d_B$-dimensional system.
We consider a (projective) unitary group representation of a group $G$ on
$\mathcal{H}_A$.
That is, for each $g \in G$, we have a unitary operator $U_g$ on $\mathcal{H}_A$ such that
for $g,g'\in G$, there exists a unit complex number $c(g,g')$ satisfying
\begin{align}
U_g U_{g'}= c(g,g') U_{gg'}.
\end{align}
When $c(g,g')$ is $1$, $\{U_g\}_{g\in G}$ is called a unitary group representation.
Otherwise, it is called a projective unitary group representation.
We also define the group twirling channel ${\cal G}(\rho):= \sum_{g \in G}\frac{1}{\vert G\vert}U_g \rho U_g^\dagger$.

Consider the case when encoding operation is given as an application of unitary.
In practical scenario, due to the device condition, 
available unitaries are limited to a subset of unitaries.
Since any combination of available unitaries is also available, 
it is natural to assume that the set of available unitaries is given as a (projective) unitary group representation of a group $G$ on
$\mathcal{H}_A$. 
In this scenario, we cannot exclude the case that $c(g,g')$ takes a non-identical element in general.
For example, when the $X$ and $Z$ operations are available on 
a qubit system,
$c(g,g')$ may take a non-identical element.
Thus, our problem is formulated as channel coding for 
the cq-channel $g (\in G)\mapsto 
U_g \rho_{AB} U_g^\dagger$.
Since this cq-channel has a group covariant form,
the channel capacity $C_{\rm c}(\rho_{AB})$ of this channel is calculated as~\cite{Holevo1998capacity,schumacher1997sending}
\begin{align}
C_{\rm c}(\rho_{AB})
=&
\sup_{p}
H\left( \sum_{g \in G} p(g)  U_g \rho_{AB} {U_g}^\dagger
\right)
-\sum_{g \in G} p(g)
H(U_g  \rho_{AB} {U_g}^\dagger) \nonumber \\
=&
\sup_{p}
H\left( \sum_{g \in G} p(g)  U_g \rho_{AB} {U_g}^\dagger
\right)
-
H( \rho_{AB} )\nonumber \\
=&
H( {\cal G} (\rho_{AB}))-H( \rho_{AB}) \nonumber \\
=& D(\rho_{AB} \| {\cal G} (\rho_{AB}))
\Label{E21},
\end{align}
where $H(\rho):= -\tr \rho \log \rho$, $D(\rho\|\sigma)=\tr \rho (\log \rho - \log \sigma)$.
In the following, we use the word ``capacity'' to denote the channel capacity.
Here, we used the following relation.
\begin{align*}
&
H\left( \sum_{g \in G} p(g)  U_g \rho_{AB} {U_g}^\dagger
\right)
=
\sum_{g'\in G} 
\frac{1}{\vert G\vert}
H\left( \sum_{g \in G} 
p(U_{g'} g U_{g'}^\dagger)  U_g \rho_{AB} {U_g}^\dagger
\right) \\
\le &
H\left( 
\sum_{g'\in G} 
\frac{1}{\vert G\vert}
\sum_{g \in G} 
p(U_{g'} g U_{g'}^\dagger)  U_g \rho_{AB} {U_g}^\dagger
\right) 
= 
H\left( 
\sum_{g \in G} 
\frac{1}{\vert G\vert}
U_g \rho_{AB} {U_g}^\dagger
\right) .
\end{align*}

We describe here the first case with quantum memory mentioned above.
In this case, Alice transmits a message to Bob with $n$ uses of this cq-channel.
The set of available encoding operations is characterized by the product group
$G^n=\overbrace{G\times \cdots \times G}^n$.
The $n$-tensor product representation of $U$ is given as
$U_{(g_1, \ldots, g_n)}:= U_{g_1} \otimes \cdots \otimes U_{g_n}$ for $(g_1, \ldots, g_n) \in G^n$.
By using the $n$-tensor product representation,
a classical message $k\in\mathcal{K}_n$ is encoded by using an encoder
\begin{equation}
\phi_{\mathrm{e},n}: K_n \to G^n
\end{equation}
and applying $U_{\phi_{\mathrm{e},n}(k)}$ to her local copies and sending them to Bob, thus producing the state
\begin{equation}
U_{\phi_{\mathrm{e},n}(k)} \rho_{AB} U_{\phi_{\mathrm{e},n}(k)}^\dagger
\end{equation}
at Bob, effectively producing the cq channel from Alice to Bob.
Using quantum memories to store the whole state, Bob then decodes the message using a decoder POVM
\begin{equation}
\phi_{\mathrm{d},n}= \{ \Pi_k\}_{k\in\mathcal{K}_n}
\end{equation}
on $\mathcal{H}_A \otimes \mathcal{H}_B$.


A pair of an encoder $\phi_{\mathrm{e},n}$ and a decoder $\phi_{\mathrm{d},n}$ is called
a code $\Phi_n$.
The performance of a code $\Phi_n$ is evaluated by the size of $\mathcal{K}_n$ denoted by 
$\vert\mathcal{K}_n\vert$
and the averaged decoding error probability $\epsilon(\Phi_n)$ given as
\begin{align}
\epsilon(\Phi_n):= \frac{1}{\vert\mathcal{K}_n\vert} \sum_{k \in {\cal K}_n}
\Big(1-\tr U_{\phi_{\mathrm{e},n}(k) }\rho_{AB}^{\otimes n}U_{\phi_{\mathrm{e},n}(k) }^\dagger \Pi_{k}
\Big).
\end{align}
Then, the channel coding theorem for cq-channels \cite{Holevo,Schumacher} 
states that the quantity in \eqref{E21} gives the operational capacity, i.e.,
\begin{align}
C_{\rm c}(\rho_{AB})=
\sup_{\{\Phi_n\}}\Big\{ \lim_{n \to \infty}\frac{1}{n}\log \vert\mathcal{K}_n\vert \Big\vert
\lim_{n \to \infty} \epsilon(\Phi_n)=0\Big\}.
\end{align}

\subsubsection{Private capacity}

In addition, the recent paper \cite{Wu-Long-Hayashi} considers private dense coding that covers 
a risk that the eavesdropper, Eve, might intercept the transmitted state 
and thus obtain (a part of) the information for the transmitted message.
In this case, if Eve wants to hide
her attack, she applies another quantum channel to the intercepted system,
sends the output to Bob,
and keeps its environment system.
However, if she does not need to hide her attack, she does not need to send the output to Bob.
We also assume
that
Eve has the environment system of the state $\rho_{AB}$, i.e., 
the joint state among Alice, Bob, and Eve is a pure state $\rho_{ABE}$.
In this case, the following two conditions are required.
\begin{description}
\item[(1)] When Eve does not intercept the transmitted state, Bob recovers the message with a probability of almost $1$.
\item[(2)] Even when Eve intercepts the transmitted state, Eve obtains no information about the message.
\end{description}
This problem was studied in the case with a preshared maximally entangled state~\cite{Deng2003} and 
a general noisy preshared state~\cite{Wu-Long-Hayashi}.
When Alice and Bob share the maximally entangled state,
Eve obtains no information about the message even when Eve intercepts the transmitted state \cite{Deng2003}.
In this case, the conventional dense coding satisfies the above two conditions.
In the general noisy case with a general state $\rho_{AB}$,
the reference~\cite{Wu-Long-Hayashi} derived the following lower bound $\underline{C}_\mathrm{c}^\mathrm{p}$ of the private capacity $C_\mathrm{c}^\mathrm{p}$ as
\begin{align}
C_\mathrm{c}^\mathrm{p}(\rho_{AB})\geq \underline{C}_\mathrm{c}^\mathrm{p}(\rho_{AB})=
D(\rho_{AB} \| {\cal G} (\rho_{AB}))-
D(\rho_{AE} \| {\cal G} (\rho_{AE})).
\end{align} 
In particular, it is known that 
the above value gives the private capacity when 
$\rho_{AB}$ is a maximally correlated state \cite[Appendix A]{Wu-Long-Hayashi}. 

\subsection{Notations for group representation}\Label{S2-2}
For further study on the capacity $C_{\rm c}(\rho_{AB})$ under the given symmetry, 
we prepare several notations for group representation.
We let $\hat{G}$ denote a set of irreducible projective unitary representations of $G$.
For all irreducible representations $\lambda\in\hat G$, let $\mathcal{H}_\lambda$ be the projective representation space and $d_\lambda$ be its dimension.
In general, the representation space ${\cal H}_A$
can be written as
\begin{align}
{\cal H}_A=\bigoplus_{\lambda \in \hat{G}'} 
{\cal H}_{\lambda} \otimes {\cal M}_\lambda,
\end{align}
where ${\cal M}_\lambda= \mathbb{C}^{n_\lambda}$
expresses the multiplicity space of the irreducible projective unitary 
presentation $\lambda$, the integer $n_\lambda$ is the multiplicity, 
and $\hat{G}'\subset \hat{G}$ is the set of irreducible representations with $n_\lambda>0$ for $\mathcal{H}_A$.
Since our representation is unitary, vectors in different irreducible components are orthogonal to each other.
Hence, we can write our space in the above way.
Therefore, 
recalling $\mathcal{H}_B=\mathbb{C}^{d_B}$, the joint system ${\cal H}_A\otimes {\cal H}_B $
is written as
\begin{align}
{\cal H}_A\otimes {\cal H}_B 
=\bigoplus_{\lambda \in 
\hat{G}'} {\cal H}_{\lambda} \otimes 
\mathbb{C}^{n_\lambda d_B}.
\end{align}

Let $\Pi_\lambda$ be the projection to ${\cal H}_{\lambda} \otimes {\cal M}_\lambda=  {\cal H}_{\lambda} \otimes 
\mathbb{C}^{n_\lambda}$.
Defining the probability distribution 
$P_\Lambda(\lambda) := \tr (\Pi_\lambda \otimes I_B) \rho_{AB}$ where we denote by $\Lambda$ the corresponding random variable and the state
$ \rho_{\lambda}:=
\frac{1}{P_\Lambda(\lambda)}
(\Pi_\lambda \otimes I_B) \rho_{AB}(\Pi_\lambda \otimes I_B)$, we have
\begin{align}
{\cal G} (\rho_{AB})=
\bigoplus_{\lambda \in \hat{G}'} P_\Lambda(\lambda)
\frac{1}{d_\lambda} I_{\lambda} \otimes 
\tr_{{\cal H}_{\lambda}} (\rho_{\lambda}),
\end{align}
where $I_B$ and $I_\lambda$ are the identity operators acting on $\mathcal{H}_B$ and $\mathcal{H}_\lambda$ respectively.
This gives 
\begin{align}
H({\cal G} (\rho_{AB}))=H(P_\Lambda)+
\sum_{\lambda \in \hat{G}'} P_\Lambda(\lambda)\Big(\log d_\lambda+ 
H(\tr_{{\cal H}_{\lambda}} \rho_{\lambda})\Big).
\end{align}
Therefore, the capacity $C_{\rm c}(\rho_{AB})$ admits the following form.
\begin{align}
C_{\rm c}(\rho_{AB})
=
H(P_\Lambda)+
\sum_{\lambda \in \hat{G}'} P_\Lambda(\lambda)\Big(\log d_\lambda+ 
H(\tr_{{\cal H}_{\lambda}} \rho_{\lambda})\Big)
-H( \rho_{AB}).
\label{E22}
\end{align}
When $\rho_{AB}$ is a pure state,
this reduces to \cite{hayashiWang}
\bal
C_{\rm c}(\rho_{AB})
&=
H(P_\Lambda)+
\sum_{\lambda \in \hat{G}'} P_\Lambda(\lambda) \Big(\log d_\lambda+ 
H(\tr_{{\cal H}_{\lambda}} \rho_{\lambda})\Big)\\
&=H(P_\Lambda)+
\sum_{\lambda \in \hat{G}'} P_\Lambda(\lambda) \Big(\log d_\lambda+ 
H(\tr_{M_{\lambda}} \rho_{\lambda})\Big)
\label{E23}
\eal
where the second equality (changing the system of the trace) follows because $\rho_{\lambda}$ is also a pure state.

To simplify our analysis,
let us introduce the following condition on a unitary representation.  

\begin{definition}[Multiplicity-free condition]\Label{assp:multiplicity-free}
We say that a unitary representation $U$ of $G$ on $\cH$ is
{\it multiplicity-free} when
there exists a subset $\hat{G}' \subset \hat{G}$
such that
the unitary representation $U$ is equivalent to
\begin{align}
\bigoplus_{\lambda \in \hat{G}'} \cH_{\lambda},
\Label{M3}
\end{align}
namely, when all multiplicities are either 0 or 1, and $\hat{G}'$ is the set of all irreducible representations such that $n_\lambda=1$.
\end{definition}

\subsection{Dense coding without quantum memory}\Label{S2-3}
\subsubsection{Standard capacity}\label{sec:standard capacity}
To achieve the capacity $C_{\rm c}(\rho_{AB})$, 
Bob needs to keep his state in his quantum memory 
while first $\mathcal{H}_B^{\otimes n}$ and then $\mathcal{H}_A^{\otimes n}$ are collected.
Since the cost to keep quantum memory is expensive,
it is practical to replace it with classical memory by performing a projective measurement on each system $\mathcal{H}_B$ as they are received.
To this end, let ${\cal B}:=\{ \vert \base{k}\rangle\}_{k=1}^{d_B}$ be Bob's measurement basis on $\mathcal{H}_B$.
We also let ${\cal B}$ denote the measurement channel with respect to the basis $\mathcal{B}$.
When Bob applies the measurement ${\cal B}$, 
the initial preshared state is given by
\bal
{\cal B}(\rho_{AB}):=
\sum_{k} \langle \base{k} \vert \rho_{AB}\vert \base{k}\rangle \otimes \vert \base{k} \rangle \langle \base{k}\vert = \sum_k P_K(k)\rho_{A\vert k}, 
\eal
where 
\bal
P_K(k)&:= \tr_A \langle \base{k} \vert \rho_{AB}\vert \base{k}\rangle\\
\rho_{A\vert k}&:=\frac{1}{P_K(k)}  
\langle \base{k} \vert \rho_{AB}\vert \base{k}\rangle\otimes \dm{\base{k}}
\eal

When Bob makes a measurement with the basis ${\cal B}$ 
on every quantum system ${\cal H}_B$
in advance,
the capacity of \eqref{E21} satisfies (notice that $\mathcal{G}$ acts on $A$ and thus commutes with $B$)
\begin{align}
C_{\rm c}({\cal B}(\rho_{AB}))
=
D({\cal B}(\rho_{AB})\| {\cal G}\circ {\cal B}(\rho_{AB}))
=
D({\cal B}(\rho_{AB})\|  {\cal B}\circ {\cal G}(\rho_{AB})).
\end{align}
Here, since the support of  
${\cal B}\circ {\cal G}(\rho_{AB})$
contains the support of 
${\cal B}(\rho_{AB})$,
the above quantity does not diverge.

Indeed, even when Bob sends his measurement outcome 
to Alice before Alice sends her message,
the capacity has the same value 
because it does not distinguish whether $k$ is at Alice or Bob, as can be seen in the following. 
We have the capacity of this setting as
\begin{align}
\sum_{k=1}^{d_B} 
P_K(k)
C_{\rm c}(\rho_{A\vert k})=&
\sum_{k=1}^{d_B} 
P_K(k)
(H( {\cal G} (\rho_{A\vert k}))-H( \rho_{A\vert k})) \nonumber \\
=&
H( {\cal G} ({\cal B}(\rho_{AB})))-H( {\cal B}(\rho_{AB}))\nonumber\\
=&D({\cal B}(\rho_{AB})\| {\cal G}\circ {\cal B}(\rho_{AB}))\nonumber\\
=& C_{\rm c}({\cal B}(\rho_{AB})).
\end{align}

The data-processing inequality for the relative entropy implies that for all measurements $\mathcal{B}$
\begin{align}
C_{\rm c}(\rho_{AB})
=
D(\rho_{AB}\|  {\cal G}(\rho_{AB}))
\ge 
D({\cal B}(\rho_{AB})\| {\cal B}\circ {\cal G}(\rho_{AB}))
=
C_{\rm c}({\cal B}(\rho_{AB}))
\Label{XL1}.
\end{align}
When the equality in \eqref{XL1} holds with a measurement $\mathcal{B}$,
Bob's quantum memory is useless with respect to the measurement $\mathcal{B}$. 
We call such a state $\rho_{AB}$ 
{\it ${\cal B}$-q-memory useless}.
In particular, 
when there exists such a measurement $\mathcal{B}$,
we call such a state $\rho_{AB}$ 
{\it q-memory useless}.
Otherwise, we call the state $\rho_{AB}$ {\it q-memory useful}.
Our aim is to characterize the condition when the equality holds. 

Here, it is instructive to clarify the difference between the setting in Ref.~\cite{hayashiWang} and ours.
Ref.~\cite{hayashiWang} discussed the channel capacity when 
Bob is allowed to make a measurement across $n$-tensor product space 
$\mathcal{H}_{B}^{\otimes n}$ before Bob receives the transmitted systems.
In contrast, we allow Bob to apply a measurement 
only on a single system $\mathcal{H}_{B}$ $n$ times.
That is, in our setting, Bob is not allowed to apply 
quantum measurements over multiple copies of preshared states on Bob's side before the communication. 
The measurement in Ref.~\cite{hayashiWang} is called 
a collective measurement and our measurement here is called 
an individual measurement.
Ref.~\cite{hayashiWang} studied the setting with a collective measurement at the asymptotic limit $n \to \infty$.
They showed that the optimal transmission rate under their setting becomes $C_{\rm c}(\rho_{AB})$ when the group $G$ is abelian, the representation $U$ is multiplicity-free, and the state $\rho_{AB}$ is pure.
Although their setting does not ask Bob to keep the quantum system 
before receiving the other parts from Alice, it still requires quantum memory in Bob's system.
This is because when Bob receives these $n$ quantum systems one by one, he needs to store these $n$ quantum systems before making a measurement across his $n$ quantum systems 
$\mathcal{H}_{B}^{\otimes n}$. 
On the other hand, our setting described above prohibits the use of quantum memory. 

\begin{remark}\label{Rem1}
In our setting, we allow a collective measurement across many quantum systems received by Alice at the decoding stage, although a collective measurement is not allowed on shared quantum systems distributed to Bob.
This assumption is justified by considering the situation where the communication from Alice to Bob is done within a short period of time while it takes much longer for quantum systems to be distributed to Bob before the communication.
Indeed, if the use of the quantum communication channel is expensive, this scenario is reasonable for the following reasons.
The distribution step (before the communication) can take a longer time because it is a preparation for the communication between Alice and Bob.
Since fast communication channels are expensive, it is reasonable to use slower channels for distributing shared quantum systems.
In such a case, to make a collective measurement on shared quantum systems distributed to Bob,
Bob needs to have a quantum memory to keep his quantum systems for a long time.
On the other hand, the communication from Alice to Bob typically needs to be accomplished quickly, where they are expected to use faster communication channels. 
There, Bob does not need to hold a good quantum memory lasting for a long time to make a collective measurement over the quantum systems received in a short period of time from Alice.  
\end{remark}

\subsubsection{Private capacity}
Next, we discuss the private dense coding
when Bob makes a measurement with the basis ${\cal B}$.
We have the following lower bound for the private capacity~\cite{Wu-Long-Hayashi}
\begin{align}
\underline{C}_{\rm c}^{\rm p}(\rho_{AB},{\cal B})=
D({\cal B}(\rho_{AB}) \| {\cal G} \circ {\cal B}(\rho_{AB}))-
D(\rho_{AE} \| {\cal G} (\rho_{AE})).
\end{align} 
The difference between 
$\underline{C}_{\rm c}^{\rm p}(\rho_{AB})$ and $\underline{C}_{\rm c}^{\rm p}(\rho_{AB},{\cal B})$
is given by
\begin{align}
\underline{C}_{\rm c}^{\rm p}(\rho_{AB})-\underline{C}_{\rm c}^{\rm p}(\rho_{AB},{\cal B})
&=
D(\rho_{AB}\|  {\cal G}(\rho_{AB}))
-D({\cal B}(\rho_{AB})\| {\cal B}\circ {\cal G}(\rho_{AB}))\nonumber \\
&=
C_{\rm c}(\rho_{AB})
-C_{\rm c}({\cal B}(\rho_{AB})) \ge 0 \Label{XL1B}.
\end{align}
Therefore, 
the equality condition in \eqref{XL1} 
is useful for comparing the lower bounds for these two private capacities.

\section{General preshared states}\Label{S3}
\subsection{Conditions for q-memory uselessness}\Label{S3-1}
We begin our investigation 
of conditions for q-memory uselessness
with a general mixed preshared resource state. 
We denote with $\rho^{-1}$ the pseudoinverse of an operator $\rho$ on $\mathcal{H}$ and its support $\mathcal{H}_\rho$.
For the general case, we consider $\left[{\cal G}\circ {\cal B}(\rho_{AB})\right]^{-1}$ of
${\cal G}\circ {\cal B}(\rho_{AB})$ and denote the support of 
the operator ${\cal G}\circ {\cal B}(\rho_{AB})$ as
${\cal H}_{{\cal B},\mathcal{G},\rho_{AB}}$. 
Then, the supports of ${\cal B}(\rho_{AB})$,
${\cal G}(\rho_{AB})$, and $\rho_{AB} $
are included in ${\cal H}_{{\cal B},\mathcal{G},\rho_{AB}}$.
Applying the equality condition for the information processing inequality for relative entropy \cite{Petz} to \eqref{XL1}, 
we have the following characterizations.

\begin{theorem}\Label{TH1}
Given a preshared state $\rho_{AB}$ and the basis ${\cal B}$, 
the following conditions are equivalent.
\begin{description}
\item[(A1)]
The equality in \eqref{XL1} holds, i.e., $\rho_{AB}$ is $\mathcal{B}$-q-memory useless.
\item[(A2)]
The following relation holds as an operator on ${\cal H}_{{\cal B},\mathcal{G},\rho_{AB}}$.
\begin{align}
\rho_{AB}= \sqrt{ {\cal G}(\rho_{AB})} 
\sqrt{{\cal G}\circ {\cal B}(\rho_{AB})}^{-1}
{\cal B}(\rho_{AB})
\sqrt{{\cal G}\circ {\cal B}(\rho_{AB})}^{-1}
\sqrt{ {\cal G}(\rho_{AB})} .
\Label{eq:equality condition general}
\end{align}
\item[(A3)]
There exists a CPTP map $\Gamma$ such that
\begin{align}
\Gamma({\cal B}(U_g \rho_{AB} U_g^\dagger))=U_g\rho_{AB}
U_g^\dagger
\end{align}
for $g \in G$.
\end{description}
\end{theorem}

\begin{proof}
The equivalence between (A1) and (A2) follows from \cite[Theorem 3]{Hayden}.
The relation (A3) $\Rightarrow$ (A1) is trivial. 
Assume (A2). 
We have $U_g{\cal G}(\rho_{AB})={\cal G}(\rho_{AB})U_g$
and $U_g{\cal G}\circ {\cal B}(\rho_{AB})={\cal G}\circ {\cal B}(\rho_{AB})U_g$.
We denote the space orthogonal to
${\cal H}_{{\cal B},\mathcal{G},\rho_{AB}}$ by 
${\cal H}_{{\cal B},\rho_{AB}}^{\perp}$.
We denote the projection to
${\cal H}_{{\cal B},\rho_{AB}}^{\perp}$ by $I^\perp$.
We consider $\sqrt{ {\cal G}(\rho_{AB})} 
\sqrt{{\cal G}\circ {\cal B}(\rho_{AB})}^{-1}$
as a map on ${\cal H}_{{\cal B},\mathcal{G},\rho_{AB}}$.
Then, we can define 
the map $T:=
\sqrt{ {\cal G}(\rho_{AB})} 
\sqrt{{\cal G}\circ {\cal B}(\rho_{AB})}^{-1}
\oplus I^\perp
$ on ${\cal H}$.

Since $
U_g\sqrt{ {\cal G}(\rho_{AB})} U_g^{-1}
=\sqrt{ {\cal G}(\rho_{AB})} 
$ and
$U_g\sqrt{{\cal G}\circ {\cal B}(\rho_{AB})}^{-1}U_g^{-1}
=\sqrt{{\cal G}\circ {\cal B}(\rho_{AB})}^{-1}$, we have
$
\sqrt{ {\cal G}(\rho_{AB})} 
\sqrt{{\cal G}\circ {\cal B}(\rho_{AB})}^{-1}
U_g=U_g
\sqrt{ {\cal G}(\rho_{AB})} 
\sqrt{{\cal G}\circ {\cal B}(\rho_{AB})}^{-1}$.
Thus, we have
\begin{align}
&T U_g {\cal B}(\rho_{AB})U_g^\dagger
T^\dagger\nonumber \\
&=\sqrt{ {\cal G}(\rho_{AB})} 
\sqrt{{\cal G}\circ {\cal B}(\rho_{AB})}^{-1}
U_g {\cal B}(\rho_{AB})U_g^\dagger 
\sqrt{{\cal G}\circ {\cal B}(\rho_{AB})}^{-1}
\sqrt{ {\cal G}(\rho_{AB})} \nonumber \\
&=U_g \sqrt{ {\cal G}(\rho_{AB})} 
\sqrt{{\cal G}\circ {\cal B}(\rho_{AB})}^{-1}
{\cal B}(\rho_{AB})
\sqrt{{\cal G}\circ {\cal B}(\rho_{AB})}^{-1}
 \sqrt{ {\cal G}(\rho_{AB})} U_g^\dagger
\nonumber \\
&=U_g \rho_{AB} U_g^\dagger,
\end{align}
where the final equation follows from the condition (A2).
Defining $\Gamma (\rho):= T\rho T^\dagger$, we have 
(A3).
\end{proof}

We remark that the above results can be extended to a general quantum channel $\mathcal{B}$ following the same proof.
Apart from \eqref{eq:equality condition general}, we obtain an alternative necessary and sufficient condition for a general resource state $\rho_{AB}$ using the result in Ref.~\cite{Hayden}. 

\begin{theorem}\Label{thm:measurement useless}
    Let $\tilde B$ be a system isomorphic to $B$ and ${\rm CX}_{B\rightarrow\tilde B}\coloneqq \sum_{k,k'=1}^{d_B} \dm{\base{k}}_B\otimes \ketbra{\base{k'+k}}{\base{k'}}$, 
    where $k'+k$ is implicitly $k'+k\ \mathrm{mod}\ d_B$, be the CNOT gate with respect to the basis $\mathcal{B}$ controlled on $B$.
    Let $\rho'_{AB\tilde B}$ be the state on $AB\tilde B$ defined as 
    \begin{equation}\begin{aligned}
     \rho'_{AB\tilde B} \coloneqq {\rm CX}_{B\rightarrow\tilde B}\,(\rho_{AB}\otimes\dm{\base{0}}_{\tilde B})\,{\rm CX}_{B\rightarrow\tilde B}^\dagger.
     \Label{eq:CNOT resource}
    \end{aligned}\end{equation} 
Then, the condition $C_{\rm c}(\rho_{AB})=C_{\rm c}(\mathcal{B}(\rho_{AB}))$ holds if and only if 
there exist the following six items satisfying the condition \eqref{eq:diagonalize necessary sufficient}.
\begin{description}
\item[1)] an index set ${\cal J}$, which is possibly an infinite set,
\item[2)] 
a decomposition $\mathcal{H}_A\otimes\mathcal{H}_B=\bigoplus_{j \in {\cal J}} \mathcal{H}_{L_j}\otimes\mathcal{H}_{R_j}$ for some subspaces $\mathcal{H}_{L_j}$ and $\mathcal{H}_{R_j}$,
\item[3)] states 
    $\eta_{L_j,g}$ on $\mathcal{H}_{L_j}$ that can depend on $g$, 
\item[4)]
states $\eta_{R_j\tilde B}$ on $\mathcal{H}_{R_j}\otimes\mathcal{H}_{\tilde B}$ independent of $g$,
\item[5)]
a probability distribution $\{p_{j\vert g}\}_j$ for each $g\in G$,
\item[6)]
a unitary $V_{AB}$ on $AB$.
\end{description}

     \begin{equation}\begin{aligned}
U_g\rho_{AB\tilde B}'U_g^\dagger = V_{AB}\left(\bigoplus_j p_{j\vert g} \eta_{L_j,g}\otimes\eta_{R_j \tilde B}\right)V_{AB}^\dagger,\ \forall g\in G.
\Label{eq:diagonalize necessary sufficient}
\end{aligned}
\end{equation}
\end{theorem}

\begin{figure}[htbp]
    \centering
\includegraphics[width=0.7\columnwidth]{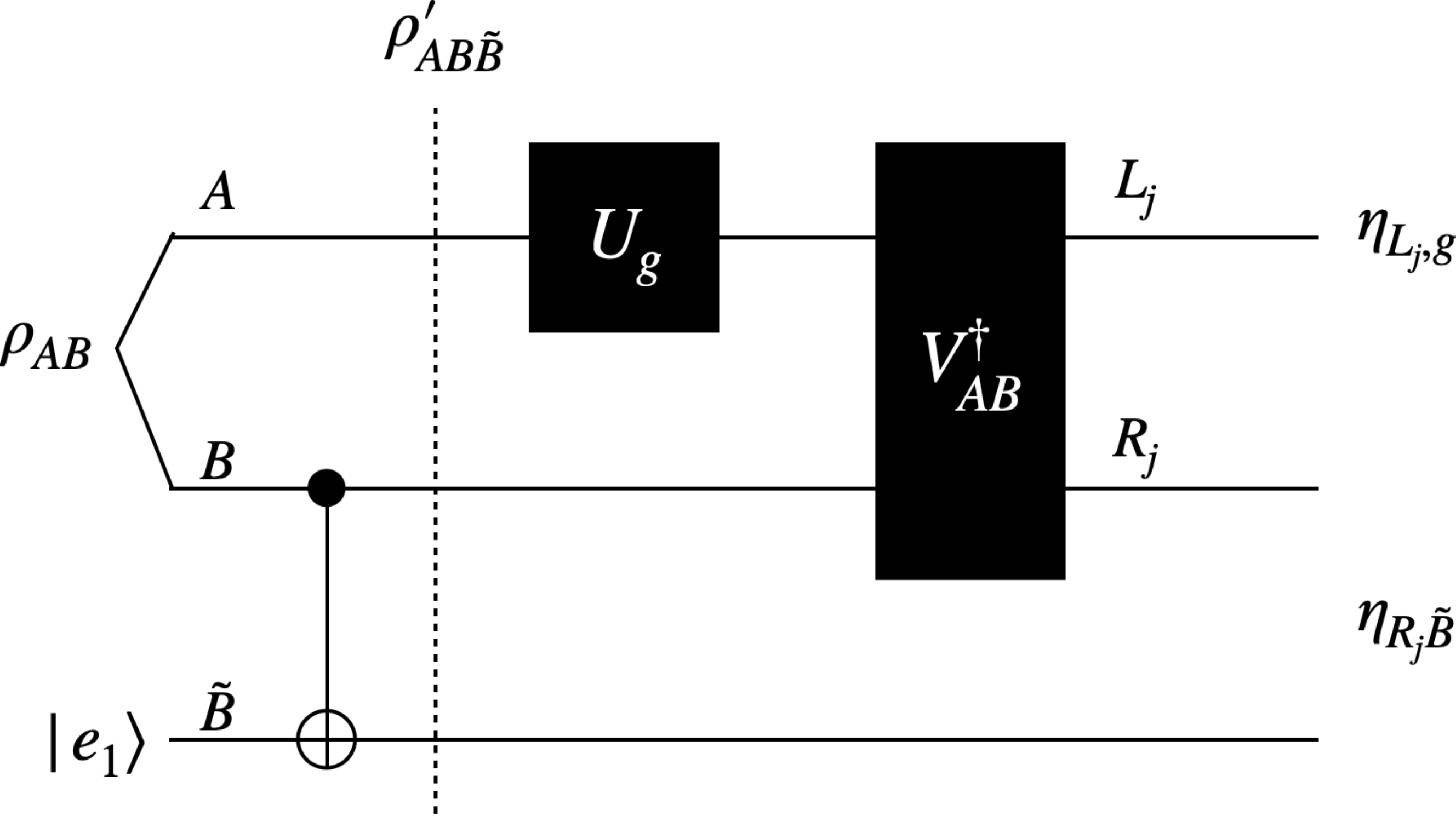}
    \caption{Circuit representation of the conditions in Theorem~\ref{thm:measurement useless}.}
    \Label{fig:necessary-sufficient}
\end{figure}

\begin{proof}
\textbf{Step 1.} 
First, we introduce another useful necessary and sufficient condition for 
the equation $C_{\rm c}(\rho_{AB})=C_{\rm c}(\mB(\rho_{AB}))$.
Define the classical-quantum state
\bal
 \tau_{GAB}&\coloneqq \sum_{g\in G}\frac{1}{\vert G\vert }\dm{g}\otimes U_g\rho_{AB} U_g^\dagger.
\eal
To take into account the action of Bob's measurement, we also consider a classical-quantum state 
\bal
 \tau_{GAB\tilde B}'&\coloneqq {\rm CX}_{B\rightarrow \tilde B} \, \tau_{GAB}\otimes\dm{\base{0}}_{\tilde B} {\rm CX}_{B\rightarrow \tilde B}^\dagger
 =\sum_{g\in G}\frac{1}{\vert G\vert }\dm{g}\otimes U_g\rho_{AB\tilde B}' U_g^\dagger,
 \Label{eq:classical quantum final state meas}
\eal
where $\rho_{AB\tilde B}'$ is the state defined in \eqref{eq:CNOT resource}.
This gives 
\bal
 \tau_{GAB}'&\coloneqq\tr_{\tilde B}(\tau'_{GAB\tilde B})
=\sum_{g\in G}\frac{1}{\vert G\vert }\dm{g}\otimes U_g\mB(\rho_{AB}) U_g^\dagger.
\Label{eq:dephased}
\eal
Note that due to the invariance of the relative entropy under isometries like in \eqref{eq:classical quantum final state meas}, we have
\bal
 C_{\rm c}(\rho_{AB})&=D(\tau_{GAB}\|\tau_G\otimes\tau_{AB})
 =D(\tau_{GAB\tilde B}'\|\tau'_G\otimes\tau'_{AB\tilde B}).
\eal
On the other hand, Eq.~\eqref{eq:dephased} implies that
\bal
 C_{\rm c}(\mB(\rho_{AB}))=D(\tau'_{GAB}\|\tau_G'\otimes\tau'_{AB}).
\eal
Therefore, $C_{\rm c}(\rho_{AB})=C_{\rm c}(\mB(\rho_{AB}))$ if and only if 
\bal
 D(\tau'_{GAB\tilde B}\|\tau_G'\otimes\tau_{AB\tilde B}')=D(\tau'_{GAB}\|\tau_G'\otimes \tau_{AB}').
 \Label{eq:relative entropy equality meas}
\eal

\textbf{Step 2.} 
Next, 
employing condition \eqref{eq:relative entropy equality meas},
we show that the existence of the six items is equivalent to the equation $C_{\rm c}(\rho_{AB})=C_{\rm c}(\mB(\rho_{AB}))$.
Applying the result from Ref.~\cite[Theorem 6]{Hayden}, we get that \eqref{eq:relative entropy equality meas} holds if and only if there exists a decomposition of Hilbert space $\mH_A\otimes\mH_B$ as 
\bal
 \mH_A\otimes\mH_B=\bigoplus_j \mH_{L_j}\otimes \mH_{R_j}
 \label{NVCT}
\eal
such that 
\bal
\tau_{GAB\tilde B}' = V_{AB}\left(\bigoplus_j q_j \eta_{GL_j}\otimes\eta_{R_j \tilde B}\right)V_{AB}^\dagger
\Label{eq:decomposition state meas}
\eal
for some unitary $V_{AB}$ on $AB$ and for some states $\eta_{GL_j}$ on $\mH_G\otimes\mH_{L_j}$ and $\eta_{R_j\tilde B}$ on $\mH_{R_j}\otimes \mH_{\tilde B}$. 
Since Eq.~\eqref{NVCT} corresponds to 
the items 1) and 2),
it is sufficient to show the equivalence between 
\eqref{eq:decomposition state meas} and 
the condition \eqref{eq:diagonalize necessary sufficient}
with the items 3), 4), 5), and 6), which does not come with the subsystem $G$ holding a classical state. 

Suppose $\tau_{GAB\tilde B}'$ has the form \eqref{eq:decomposition state meas}.
Because of \eqref{eq:classical quantum final state meas}, we have 
\bal
\sum_{g\in G} \dm{g}\tau_{GAB\tilde B}'\dm{g}=\tau_{GAB\tilde B}'.
\eal
This implies that 
\bal
\sum_{g\in G}\dm{g}\eta_{GL_j}\dm{g}=\eta_{GL_j}\ \forall j.
\eal
This condition requires $\eta_{GL_j}$ to have classical states on the subsystem $G$, which has the classical-quantum form
\bal
 \eta_{GL_j}=\sum_{g\in G}\dm{g}\otimes c_{j,g}\, \eta_{L_j, g},\quad \forall j
 \Label{eq:left state}
\eal
for some coefficient $c_{j,g}$ and some state $\eta_{L_j,g}$ on $\mH_{L_j}$.
Since $\tr\bra{g}\tau_{GAB\tilde B}'\ket{g}=1/\vert G\vert $ by \eqref{eq:classical quantum final state meas}, we have $\sum_j q_j c_{j,g}=1/\vert G\vert ,\ \forall g\in G$. 
Together with the form in \eqref{eq:classical quantum final state meas}, we get 
\bal
 U_g\rho_{AB\tilde B}'U_g^\dagger=V_{AB}\left(\bigoplus_j p_{j\vert g} \eta_{L_j,g}\otimes \eta_{R_j \tilde B}\right)V_{AB}^\dagger,\quad \forall g\in G
 \Label{eq:decomposition intermediate}
\eal
where $p_{j\vert g}\coloneqq \vert G\vert \,q_jc_{j,g}$ constructs a probability distribution labeled by $g$ with $\sum_j p_{j\vert g}=1$. 
Since this is the form advertised in \eqref{eq:diagonalize necessary sufficient}, we just showed that \eqref{eq:decomposition state meas} implies \eqref{eq:diagonalize necessary sufficient}.

On the other hand, one can recover the form of \eqref{eq:decomposition state meas} by applying $\frac{1}{\vert G \vert}\sum_g \dm{g}\otimes$ from left in both sides of \eqref{eq:diagonalize necessary sufficient}, which means that \eqref{eq:diagonalize necessary sufficient} implies \eqref{eq:decomposition state meas}.
Therefore, \eqref{eq:diagonalize necessary sufficient} and \eqref{eq:decomposition state meas} are equivalent, which concludes the proof.
\end{proof}

Let us consider a mixed-state example that satisfies the conditions in 
Theorem~\ref{thm:measurement useless}. 
Clearly, 
when the state $\rho_{AB}$ has the following separable form
\begin{align}
\rho_{AB}= \sum_{k=1}^{d_B} P_K(k) 
\vert \psi_{k}\rangle \langle \psi_{k}\vert
 \otimes \vert \base{k}\rangle \langle \base{k}\vert,
\Label{ZCXOT}
\end{align}
we have ${\cal B}(\rho_{AB})=\rho_{AB}$, which implies 
the equality of the capacities in \eqref{XL1}.
That is, the state $\rho_{AB}$ is ${\cal B}$-q-memory useless.
However, a general separable state cannot be written in the form \eqref{ZCXOT}
even when the basis $\{\ket{\base{k}}\}_{k=1}^{d_B}$ can be freely chosen
because a general separable state is written with the sum of non-orthogonal product states.
Hence, it is, in general, unclear whether a general separable state $\rho_{AB}$ is q-memory useless.

Let us now investigate whether there exists an entangled q-memory useless state.

\begin{proposition}\Label{LLA}
Let $\rho_{AB}$ be a state having the following form
\begin{align}
\rho_{AB}= \sum_{j}P_J(j) \vert \Psi_j\rangle \langle \Psi_j\vert ,\Label{ZCXO}
\end{align}
where the state $\vert \Psi_j\rangle$ has the form 
\begin{align}
\vert \Psi_j\rangle
=\sum_{k=1}^{d_B}
e^{i \theta_{k,j}}
\sqrt{P_{K\vert J}(k\vert j)}
\sum_{\lambda\in \hat{G}'}
e^{i \theta_{\lambda,k}}
\sqrt{P_{\Lambda}(\lambda)}
\vert \psi_{\lambda,k}\rangle \otimes \vert \base{k}\rangle,
\Label{eq:equalityB}
\end{align}
and the vector $\vert \psi_{\lambda,k}\rangle \in  {\cal H}_{\lambda} \otimes {\cal M}_\lambda$ satisfies
\begin{align}
\tr_{{\cal M}_{\lambda}} \vert \psi_{\lambda,k}\rangle\langle \psi_{\lambda,k}\vert 
=\tr_{{\cal M}_{\lambda}} \vert \psi_{\lambda,k'}\rangle\langle \psi_{\lambda,k'}\vert 
\Label{AC8B}
\end{align}
for every $k,k'$.
Then, there is a basis $\mathcal{B}$ such that the conditions for Theorem~\ref{thm:measurement useless} hold.
That is, the state $\rho_{AB}$ is q-memory useless.
\end{proposition}

\begin{proof}
\textbf{Step 1.} 
First, we observe that for each $\ket{\psi_{\lambda,k}}$ in \eqref{eq:equalityB} there exists a Schmidt decomposition
\begin{align}
    \ket{\psi_{\lambda,k}} = \sum_j \sqrt{c_{\lambda, j}}\,\ket{r_{\lambda, j}}_{\cal H_\lambda}\ket{m_{\lambda, k, j}}_{\mathcal{M}_\lambda}
    \Label{eq:Alice pure state}
\end{align} 
where neither $c_{\lambda, j}$ nor $\ket{r_{\lambda, j}}$ depends on $k$.
This is because the marginal states for $\ket{\psi_{\lambda,k}}$ on $\mathcal{H}_\lambda$ do not depend on $k$ because of \eqref{AC8B}, and thus we can use the same eigenbasis for $\mathcal{H}_\lambda$ in the Schmidt decomposition, recalling that any eigenbasis of the marginal state gives a Schmidt basis.

\textbf{Step 2.}
Next, employing the form in \eqref{eq:Alice pure state}, we show that $\rho_{AB}$ in \eqref{ZCXO} satisfies the conditions in Theorem~\ref{thm:measurement useless}. 
Let us define $\vert \Psi'_j\rangle$ as
\begin{align}
\vert \Psi'_j\rangle:= {\rm CX}_{B\rightarrow\tilde B} \vert \Psi_j\rangle\ket{\base{0}}_{\tilde B},
\end{align}
where $\ket{\base{0}}_{\tilde B}$ is a computational basis state on $\tilde B$.
Then, using \eqref{eq:equalityB} and \eqref{eq:Alice pure state}, we can write
\begin{align}
 \ket{\Psi'_j}=& \sum_\lambda \sqrt{P_\Lambda(\lambda)} 
 \sum_{k=1}^{d_B} \sqrt{P_{K\vert J}(k\vert j)}
 e^{i(\theta_{\lambda,k}+\theta_{k,j})} 
    \sum_l \sqrt{c_{\lambda, l}}\ket{r_{\lambda, l}}_{\cal H_\lambda}
    \ket{m_{\lambda, k, l}}_{\mathcal{M}_\lambda}\ket{\base{k}}_B\ket{\base{k}}_{\tilde B}.
\end{align}
Applying $U_g=\bigoplus_\lambda U_{\lambda, g}\otimes I_{\mathcal{M}_\lambda}$ to $\ket{\Psi'}$, we get 
\begin{align}
 &U_g\ket{\Psi'_j } \nonumber \\
 =& \sum_\lambda \sqrt{P_\Lambda(\lambda)} \sum_{k=1}^{d_B} \sqrt{P_{K\vert J}(k\vert j)}
 e^{i(\theta_{\lambda,k}+\theta_{k,j})} 
  \sum_l \sqrt{c_{\lambda, l}}\ket{r_{\lambda, g, l}}_{\cal H_\lambda}
    \ket{m_{\lambda, k, l}}_{\mathcal{M}_\lambda}\ket{\base{k}}_B\ket{\base{k}}_{\tilde B}
\end{align}
where $\ket{r_{\lambda, g, l}}\coloneqq U_{\lambda, g}\ket{r_{\lambda, l}}$.
Let $V_k\coloneqq \bigoplus_\lambda e^{-i\theta_{\lambda,k}} I_{\cal H_\lambda}\otimes 
\sum_l \ketbra{m_{\lambda, l}}{m_{\lambda, k, l}}_{\mathcal{M}_\lambda}$ be a unitary on $A$ 
where $\{\ket{m_{\lambda, l}}\}_l$ is an arbitrary basis independent of $k$. 
We then define a controlled unitary $V_{AB}\coloneqq \sum_{k=1}^{d_B} V_k\otimes\dm{\base{k}}_B$.
We get 
\begin{align}
& V_{AB}U_g\ket{\Psi'_j}\nonumber \\
    &= \sum_\lambda \sqrt{P_\Lambda(\lambda)} \sum_{k=1}^{d_B} e^{i\theta_{k,j}}\sqrt{P_{K\vert J}(k\vert j)} \sum_l \sqrt{c_{\lambda, l}}\ket{r_{\lambda, g, l}}_{\cal H_\lambda}
    \ket{m_{\lambda, l}}_{\mathcal{M}_\lambda}\ket{\base{k}}_B\ket{\base{k}}_{\tilde B}\nonumber \\
    &=\left(\sum_\lambda \sqrt{P_\Lambda(\lambda)} 
    \sum_l \sqrt{c_{\lambda, l}}\ket{r_{\lambda, g, l}}_{\cal H_\lambda}
    \ket{m_{\lambda, l}}_{\mathcal{M}_\lambda}\right)
    \otimes\left(\sum_{k=1}^{d_B} e^{i\theta_{k,j}} \sqrt{P_{K\vert J}(k\vert j)}\ket{\base{k}}_B\ket{\base{k}}_{\tilde B}\right)\\
    &=:\ket{\Psi_{A,g}}\otimes \ket{\Psi_{B\tilde B,j}}
\end{align}
Therefore, we have for the state $\rho_{AB}=\sum_j P_J(j)\dm{\Psi_j}$ that
\begin{equation}\begin{aligned}
 V_{AB}U_g\rho_{AB}  U_g^\dagger V_{AB}^\dagger=\vert  \Psi_{A,g}\rangle\langle \Psi_{A,g}\vert \otimes \sum_j P_J(j)\dm{\Psi_{B\tilde B,j}} .
\end{aligned}\end{equation}
This is the form in \eqref{eq:diagonalize necessary sufficient}, 
where we choose $L_j=A$, $R_j=B$ for item 2) in Theorem~\ref{thm:measurement useless}. This concludes the proof that the state $\rho_{AB}$ in \eqref{ZCXO} satisfies the conditions in Theorem~\ref{thm:measurement useless}.
\end{proof}

\subsection{Examples of entangled q-memory useless state
based on Proposition \ref{LLA}}\Label{S3-2}
We derive examples of entangled q-memory useless states
using Proposition \ref{LLA}.
This construction contains mixed states as well as pure states.
To simplify the condition in Proposition \ref{LLA},
we assume that our unitary representation $U_g$ of $G$ on $\cH$ is multiplicity-free.
Then,
the condition \eqref{eq:equalityB} is rewritten as 
\begin{align}
\vert \Psi_j\rangle
=\sum_{k=1}^{d_B} 
e^{i \theta_{k,j}}
\sqrt{P_{K\vert J}(k\vert j)}
\Big(\sum_{\lambda\in \hat{G}'}
e^{i \theta_{\lambda,k}}
\sqrt{P_{\Lambda}(\lambda)}
\vert \psi_{\lambda}\rangle \Big)\otimes \vert \base{k}\rangle.
\Label{eq:equalityB2}
\end{align}
where $\hat{G}'$ with $\vert \hat G'\vert =l$ is a subset of irreducible representations.
We write the set $\hat{G}'$ as
$\{\lambda_s\}_{s=1}^l$ and remember that 
the basis ${\cal B}$ is written as $\{ \vert \base{k}\rangle\}_{k=1}^{d_B} $.
As a special case, let ${d_B}=l$,
$P_{K\vert J}(k\vert j)=
P_{\Lambda}(\lambda_s)=\frac{1}{l}$,
$\theta_{\lambda_s,k }=\frac{2\pi s k}{l}$,
and
$\theta_{k,j }=\frac{2\pi j k}{l}$.
The states
\bal
\ket{\phi_{k}}:= \sum_{\lambda_s\in \hat{G}'}
e^{i 
\frac{2\pi s k}{l}
}
\sqrt{P_{\Lambda}(\lambda_s)}
\vert \psi_{\lambda_s}\rangle,\quad k=1, \ldots, d_B
\eal
form normalized orthogonal vectors.
The states $\vert \Psi_j\rangle$ 
 with $j=1,\ldots,l $ are simplified as
\begin{align}
\vert \Psi_j\rangle
=\sum_{k=1}^l 
e^{i \frac{2\pi j k}{l}}
\frac{1}{\sqrt{l}}
\ket{\phi_{k}}\otimes \vert \base{k}\rangle.
\Label{eq:equalityB3}
\end{align}
We find that the states $\vert \Psi_j\rangle$ with $j=1,\ldots,l $ also form normalized orthogonal vectors and are all examples of entangled q-memory useless states.
Then, the state 
$\rho_{AB}= \sum_{j=1}^lP_J(j) \vert \Psi_j\rangle \langle \Psi_j\vert $
is separable if and only if $P_J$ is the uniform distribution on $\{1, \ldots, l\}$~\cite{Rains1999bound}.
In this case, the state $\rho_{AB}$ is written as
\bal
\rho_{AB}&=\frac{1}{l}\sum_{k=1}^l \dm{\Psi_j}=\frac{1}{l}\sum_{k=1}^l
\ketbra{\phi_{k}}{\phi_{k}}\otimes \ketbra{\base{k}}{\base{k}}.
\eal

In addition, the state $\rho_{AB}$ in the example is 
a maximally correlated state with respect to the basis 
$\{\ket{\phi_{k}}\otimes \vert \base{k'}\rangle\}_{k,k'=1}^l$.
Hence, the lower bound $\underline{C}_{\rm c}^{\rm p}(\rho_{AB})$ of the private capacity
equals the private capacity in this case \cite[Appendix A]{Wu-Long-Hayashi}. 
The above example also guarantees that 
$\underline{C}_{\rm c}^{\rm p}(\rho_{AB})=\underline{C}_{\rm c}^{\rm p}(\rho_{AB},{\cal B})$.
Thus, $\underline{C}_{\rm c}^{\rm p}(\rho_{AB},{\cal B})$ also attains the private capacity.
In this case, even when Bob makes a measurement in the basis ${\cal B}$,
Alice and Bob can attain the private capacity.
That is, Bob's quantum memory is not needed even for 
private dense coding.

\subsection{Examples of entangled q-memory useless state
based on Condition (A2)}\Label{S3-3}

Let $G$ be commutative and $\mathcal{H}_A$ be multiplicity free, and consider the maximally entangled state 
\bal
\ket{\Phi}=\sum_{s=1}^l \frac{1}{\sqrt{l}} \vert\psi_{\lambda_s}\rangle 
\vert v_{\lambda_s}\rangle
\label{VG1}
\eal
where the set $\hat{G}'$ of irreducible representations appearing in 
${\cal H}_A$ is given as $\{\lambda_s\}_{s=1}^l$
and $\{\vert v_{\lambda_s}\rangle\}_{s=1}^l$ are orthogonal states on ${\cal H}_B$ defined by
\begin{align}
\vert v_{\lambda_s}\rangle= 
\frac{1}{\sqrt{l}}\sum_{k=1}^l e^{-2\pi k {\lambda_s} i/l }
\vert \base{k}\rangle.\Label{VG2}
\end{align}
Although the previous section shows the fact that
the maximally entangled state $\vert \Phi \rangle$
is q-memory useless by using Proposition \ref{LLA},
this section shows the same fact by using the condition (A2).

Note that
\begin{align}
\vert \Phi \rangle=
\frac{1}{l}\sum_{k=1}^l \sum_{s=1}^l 
e^{-2\pi k {\lambda_s} i/l }
\vert\psi_{\lambda_s}\rangle \vert \base{k}\rangle
\end{align}
Therefore, 
defining $\vert\phi_k \rangle:=
\frac{1}{\sqrt{l}}\sum_{s=1}^l 
\vert\psi_{\lambda_s}\rangle$, we obtain the following lemma.
\begin{lemma}
The maximally entangled state 
$\vert \Phi \rangle$ has the following form.
\begin{align}
\vert \Phi \rangle=
\frac{1}{\sqrt{l}}\sum_{k=1}^l 
\vert\phi_k \rangle 
\vert \base{k}\rangle.
\end{align}
    \end{lemma}
While the capacity is $\log l$, it can be achieved even when Bob measures the system $B$ 
with the basis
$\{\vert e_k\rangle\}_k$ 
because 
the resultant state on $A$ is 
$\vert\phi_k\rangle$ and 
$H({\cal G}
(\vert\phi_k\rangle \langle \phi_k\vert))=\log l$.
Hence, $\vert \Phi \rangle$
is q-memory useless. 

This q-memory useless property can be seen
from the fact that
the state 
$\vert \Phi \rangle$ satisfies the condition (A2) as follows.
We have
\begin{align}
{\cal B}(\dm{\Phi} )&=
\frac{1}{l}\sum_{k=1}^l 
\ketbra{\phi_k }{\phi_k } \otimes 
\ketbra{\base{k} }{\base{k} } \Label{NM3}\\
{\cal G}(\dm{\Phi})&=
\frac{1}{l}\sum_{s=1}^l 
\ketbra{\psi_{\lambda_s} }{\psi_{\lambda_s} } \otimes 
\ketbra{v_{\lambda_s} }{v_{\lambda_s} } \Label{NM4}\\
{\cal G}\circ {\cal B}(\dm{\Phi} )&=\frac{1}{l^2}I.\Label{NM1}
\end{align}
The matrices $l {\cal B}(\dm{\Phi} )$
and $l {\cal G}(\dm{\Phi} )$ are projections.
The intersection of the range of these projections is
the one-dimensional space generated by $\ket{ \Phi }$.
Hence, 
\begin{align}
\sqrt{ {\cal G}(\dm{\Phi} )}
 {\cal B}(\dm{\Phi} )
\sqrt{ {\cal G}(\dm{\Phi})}=\frac{1}{l^2}\dm{\Phi}.\Label{NM2}
\end{align}
The combination of \eqref{NM1} and \eqref{NM2} guarantees 
the condition (A2).
Therefore, the state $\ketbra{ \Phi }{ \Phi}$ 
is q-memory useless.

This fact can be extended to the dephased state
$\rho_{AB,p}
=(1-p)\ketbra{ \Phi }{ \Phi}+p {\cal G}(\ketbra{ \Phi }{ \Phi} )$~\cite{Devetak2005capacity}
as follows.
Instead of \eqref{NM3}, we have
\begin{align}
{\cal B}(\rho_{AB,p} )&=
(1-p) \frac{1}{l}\sum_{k=1}^l 
\ketbra{\phi_k }{\phi_k } \otimes 
\ketbra{\base{k} }{\base{k} }+ \frac{p}{l^2}I
\Label{NM3+B}.
\end{align}
Even when $\rho_{AB,0}$ is replaced by $\rho_{AB,p}$,
the relations \eqref{NM4} and \eqref{NM1} still hold as follows.
\begin{align}
{\cal G}(\rho_{AB,p})&=
(1-p)
{\cal G}(\ketbra{ \Phi }{ \Phi})+p {\cal G}(\ketbra{ \Phi }{ \Phi} )
=
\frac{1}{l}\sum_{\lambda=1}^l 
\ketbra{\psi_\lambda }{\psi_\lambda } \otimes 
\ketbra{v_\lambda }{v_\lambda } \Label{NM4Y}\\
{\cal G}\circ {\cal B}(\rho_{AB,p} )&=
(1-p) \frac{1}{l}\sum_{k=1}^l 
{\cal G}(\ketbra{\phi_k }{\phi_k }) \otimes 
\ketbra{e_k }{e_k }+ \frac{p}{l^2}I
\nonumber\\
&=
(1-p) \frac{1}{l}\sum_{k=1}^l 
I\otimes 
\ketbra{e_k }{e_k }+ \frac{p}{l^2}I
=\frac{1}{l^2}I.\Label{NM1Y}
\end{align}
Combining \eqref{NM2} and \eqref{NM3+B} with the above relations,
we have
\begin{align}
&\sqrt{ {\cal G}(\rho_{AB,p} )}
\sqrt{ {\cal G}\circ {\cal B}(\rho_{AB,p} )}^{-1}
 {\cal B}(\rho_{AB,p} )
\sqrt{ {\cal G}\circ {\cal B}(\rho_{AB,p} )}^{-1}
\sqrt{ {\cal G}(\rho_{AB,p})}\nonumber \\
=&
(1-p)\ketbra{ \Phi }{ \Phi}+p {\cal G}(\ketbra{ \Phi }{ \Phi} ).
\end{align}
Since the condition (A2) holds,
the dephased state $\rho_{AB}=(1-p)\ketbra{ \Phi }{ \Phi}+p {\cal G}(\ketbra{ \Phi }{ \Phi} )$ 
serves as another example of an entangled q-memory useless state.

\section{Pure preshared states}\Label{S4}
\subsection{General conditions for q-memory uselessness}\Label{S4-1}
Next, we consider the case with a pure state $\rho_{AB}= \vert \Psi\rangle\langle \Psi\vert $.
For ${\cal B}=\{ \vert \base{k}\rangle\}_{k=1}^{d_B} $ on $\mathcal{H}_B $,
without loss of generality, we can assume that 
$ \langle \base{k}\vert  \rho_B\vert \base{k} \rangle > 0$.
(Otherwise, we restrict $k$ to the index to satisfy 
$ \langle \base{k} \vert  \rho_B\vert \base{k} \rangle > 0$.)
Then, we have 
\begin{align}
\vert \Psi\rangle_{AB}=\sum_{k=1}^{d_B}\sqrt{P_K(k)}\vert \psi_k\rangle_A\vert \base{k}\rangle_B.
\end{align}

\begin{theorem}\Label{TH1P}
Assume that the preshared state $\rho_{AB}$
is a pure state $\vert \Psi\rangle\langle \Psi\vert $.
Given the basis ${\cal B}$, 
the following condition is equivalent to 
Conditions (A1), (A2), and (A3).
\begin{description}
\item[(C1)]
The state $\vert \Psi\rangle$ has the form as
\begin{align}
\vert \Psi\rangle
=\sum_{k=1}^{d_B} \sqrt{P_K(k)}\sum_{\lambda\in \hat{G}'}\sqrt{P_{\Lambda}(\lambda)}
e^{i \theta_{\lambda,k}}
\vert \psi_{\lambda,k}\rangle \otimes \vert \base{k}\rangle,\Label{eq:equality condition 1 pure}
\end{align}
where for all $k,k'$ the vector $\vert \psi_{\lambda,k}\rangle \in  {\cal H}_{\lambda} \otimes {\cal M}_\lambda$ 
satisfies the condition;
\begin{align}
\tr_{{\cal M}_{\lambda}} \vert \psi_{\lambda,k}\rangle\langle \psi_{\lambda,k}\vert 
=\tr_{{\cal M}_{\lambda}} \vert \psi_{\lambda,k'}\rangle\langle \psi_{\lambda,k'}\vert ,
\Label{AC8}
\end{align}
and $P_K$ and $P_\Lambda$ are independent probability distributions.
\end{description}
\end{theorem}

When a state $\rho_{AB}$ is a pure state,
Condition (C1) is equivalent to the condition in Proposition \ref{LLA}.
That is, the above theorem shows the inverse direction of 
Proposition \ref{LLA} in the case with pure states. 

\begin{proof}
It is sufficient to show 
the equivalence between (C1) and (A1).
Defining 
\bal
 \rho_{\lambda,k}\coloneqq\frac{1}{P_K(k)P_{\Lambda\vert k}(\lambda)}(P_\lambda\otimes \dm{k})\rho_{AB}(P_\lambda\otimes\dm{k})
\eal
and recalling \eqref{E22}, we have
\begin{align}
&C_{\rm c}({\cal B}(\rho_{AB}))\nonumber \\
=&
\sum_{\lambda \in \hat{G}'} P_{\Lambda}(\lambda)\log d_\lambda
+\sum_{k=1}^{d_B}
P_K(k)
\Big(H(P_{\Lambda\vert k})+
\sum_{\lambda \in \hat{G}'} P_{\Lambda\vert k}(\lambda)
H(\tr_{{\cal H}_{\lambda}} \rho_{\lambda,k})\Big) \nonumber \\
=&
\sum_{\lambda \in \hat{G}'} P_{\Lambda}(\lambda)\log d_\lambda
+\sum_{k=1}^{d_B} 
P_K(k) H(P_{\Lambda\vert k})
+\sum_{\lambda \in \hat{G}'} P_{\Lambda}(\lambda)
\sum_{k=1}^{d_B} P_{K\vert \lambda}(k)
H(\tr_{{\cal H}_{\lambda}} \rho_{\lambda,k}) \nonumber \\
=&
\sum_{\lambda \in \hat{G}'} P_{\Lambda}(\lambda)\log d_\lambda
+\sum_{k=1}^{d_B} 
P_K(k) H(P_{\Lambda\vert k})
+\sum_{\lambda \in \hat{G}'} P_{\Lambda}(\lambda)
\sum_{k=1}^{d_B} P_{K\vert \lambda}(k)
H(\tr_{M_{\lambda}} \rho_{\lambda,k}) .
\Label{ZX7}
\end{align}
The final equation follows from the fact that $\rho_{\lambda,k}$ is a pure state because $\rho_{AB}$ is pure.

In this case, 
by comparing \eqref{E23} and \eqref{ZX7},
Condition (A1), i.e., $C_{\rm c}(\rho)=C_{\rm c}(\mathcal{B}(\rho))$, is equivalent to
\begin{equation}\begin{aligned}
& H(P_\Lambda)+
\sum_{\lambda \in \hat{G}'} P_\Lambda(\lambda) 
H(\tr_{M_{\lambda}} \rho_{\lambda})\\
= &
\sum_{k=1}^{d_B} 
P_K(k) H(P_{\Lambda\vert k})
+\sum_{\lambda \in \hat{G}'} P_{\Lambda}(\lambda)
\sum_{k=1}^{d_B} P_{K\vert \lambda}(k)
H(\tr_{M_{\lambda}} \rho_{\lambda,k}). \label{XPP}
\end{aligned}\end{equation}
Since $P_{\Lambda}=  \sum_{k=1}^{d_B} 
P_K(k) P_{\Lambda\vert k}$ and $\tr_{M_{\lambda}} \rho_{\lambda}
=\sum_{k=1}^{d_B} P_{K\vert \lambda}(k)
\tr_{M_{\lambda}} \rho_{\lambda,k}$, 
the concavity of entropy implies that
\begin{align}
H(P_\Lambda) \ge &
\sum_{k=1}^{d_B} 
P_K(k) H(P_{\Lambda\vert k}) \Label{XPP1} \\
H(\tr_{M_{\lambda}} \rho_{\lambda})
\ge &
\sum_{k=1}^{d_B} P_{K\vert \lambda}(k)
H(\tr_{M_{\lambda}} \rho_{\lambda,k})\quad \forall \lambda 
 \Label{XPP2}
\end{align}
and therefore, \eqref{XPP} is equivalent to 
\begin{align}
P_{\Lambda}
=P_{\Lambda\vert k}, ~
\tr_{M_{\lambda}} \rho_{\lambda}
=
\tr_{M_{\lambda}} \rho_{\lambda,k}
\Label{ZXP2}
\end{align}
for all $k$ and $\lambda$.
Since the condition \eqref{ZXP2} 
is equivalent to Condition (C1), we obtain 
the equivalence between (C1) and (A1).
\end{proof}

Although we have already proven the equivalence between (C1) and all the conditions in Theorem~\ref{TH1}, it may not be immediately clear how the form of (C1) can induce the form in (A2). In the Appendix, we explicitly derive how the form in (A2) can be reached starting from the form in (C1).

Theorem~\ref{TH1} and the proof of Proposition~\ref{LLA} give
the following characterization for pure resource states. 
The following is a summary of the results for general pure states.
\begin{theorem}
For an arbitrary pure resource state $\rho_{AB}=\dm{\Psi}$, the following are equivalent.
\begin{description}
\item[(D1)]
The equality in \eqref{XL1} holds, i.e., $\rho_{AB}$ is $\mathcal{B}$-q-memory useless.
\item[(D2)]
$\ket{\Psi}$ satisfies \eqref{eq:equality condition 1 pure} and \eqref{AC8}.
\item[(D3)]
The following relation holds.
\begin{align}
\rho_{AB}= &\sqrt{ {\cal G}(\rho_{AB})} 
\sqrt{{\cal G}\circ {\cal B}(\rho_{AB})}^{-1}
{\cal B}(\rho_{AB})\sqrt{{\cal G}\circ {\cal B}(\rho_{AB})}^{-1}
\sqrt{ {\cal G}(\rho_{AB})}.
\end{align}
\item[(D4)] \Label{item:Hayden condition pure}
    There exists a controlled unitary $V_{AB}=\sum_{k=1}^{d_B} V_k\otimes\dm{\base{k}}_B$ and a set $\{\ket{\eta_g}_A\}_g$ of states on $\mathcal{H}_{A}$ and some state $\ket{\phi}_B$ on $\mathcal{H}_{B}$ such that
\begin{equation}\begin{aligned}
V_{AB}\,U_g\ket{\Psi}= \ket{\eta_g}_A\otimes\ket{\phi}_B, \forall g\in G.
\end{aligned}\end{equation} 
\end{description}
\end{theorem}

To obtain a further characterization, we recall the multiplicity-free condition.
When the multiplicity-free condition is satisfied, \eqref{eq:equality condition 1 pure} and \eqref{AC8} can be represented by the following single form 
\begin{align}
\vert \Psi\rangle
=\sum_{\lambda\in \hat{G}'}\sqrt{P_{\Lambda}(\lambda)}
\vert \psi_{\lambda}\rangle \otimes \sum_{k=1}^{d_B} \sqrt{P_K(k)}e^{i \theta_{\lambda,k}}\vert \base{k}\rangle,
\Label{eq:equality condition multiplicity free}
\end{align}
for some vector $\vert \psi_{\lambda}\rangle \in  {\cal H}_{\lambda}$ independent of $k$.
On the other hand, if $G$ is abelian (but not necessarily multiplicity-free), then $\dim {\cal H}_\lambda = 1$ for all $\lambda\in \hat{G}'$, implying that \eqref{AC8} is always satisfied and thus the condition \eqref{eq:equality condition 1 pure} solely serves as the necessary and sufficient condition for Bob's quantum memory to be useless.  

Therefore, we obtain the following corollary.

\begin{corollary}\Label{Cor1}
Let the preshared state $\rho_{AB}$ be a pure state given by $\ket{\Psi}=\sum_{\lambda\in\hat G'}\sqrt{P_\Lambda(\lambda)}\ket{\psi_\lambda}\ket{v_\lambda}$, where $\hat G'$ with $\vert \hat G'\vert =l$ is a subset of 
irreducible representations for which $P_\Lambda(\lambda)\neq 0\ \forall \lambda\in\hat G'$, and $\ket{v_1}, \ldots, \ket{v_l}$ are
some quantum states on $\mathcal{H}_B$.
Suppose also that the multiplicity-free condition holds.
Then, the following conditions are equivalent.

\begin{description}
\item[(E1)]
The state $\vert \Psi\rangle$
is q-memory useless.
\item[(E2)]
There exist a basis $\{\vert \base{k}\rangle\}_{k=1}^{d_B}$ of ${\cal H}_B$, a probability distribution $P_K$ on $\{1,\dots,d_B\}$ independent of $\lambda$,
and real numbers $\{\theta_{\lambda,k}\}_{\lambda,k}$ such that
\begin{align}
\vert v_\lambda\rangle= \sum_{k=1}^{d_B} \sqrt{P_K(k)}e^{i \theta_{\lambda,k}}\vert \base{k}\rangle.
\Label{eq:unitary condition}
\end{align}
\end{description}
\end{corollary}

Since Condition (E2) can be considered as a condition for 
the $l$ vectors $V :=(\ket{v_1},\dots,\ket{v_l})$ in ${\cal H}_B$,
the problem of q-memory uselessness
is reduced to the analysis for the $l$ vectors $V$.

For example, when 
the vectors $\ket{v_1}, \ldots, \ket{v_l}$ are orthogonal,
Condition (E2) holds by choosing 
the basis $\{\vert \base{k}\rangle\}_{k=1}^{d_B}$ of ${\cal H}_B$ as follows.
The orthogonality guarantees that $l \le d_B$.
We set 
$\vert \base{1}\rangle, \ldots, \vert \base{l}\rangle$ as 
\bal
 \ket{k} = \frac{1}{\sqrt{l}}\sum_{s=1}^l e^{i2\pi ks/l}\ket{v_s} ,\quad k=1,\dots,l
\eal
and 
$\vert \base{l+1}\rangle, \ldots, \vert \base{r}\rangle$
to be vectors orthogonal to
$\ket{v_1}, \ldots, \ket{v_l}$.
Eq.~\eqref{eq:unitary condition} is then satisfied by choosing $P_K(k)=1/l$ for $k=1,\dots,l$ and $\theta_{\lambda,k}=-2\pi\lambda k/l$.

The fact that Condition (E2) always holds when $l=2$ can be seen as follows.
We choose $\theta$ and $\theta'$ as $e^{i\theta'}\cos \theta= \langle v_1 \vert v_2\rangle $.
Then, we choose the basis $\{\vert \base{1}\rangle,\vert \base{2}\rangle\}$ to satisfy 
\begin{align}
\vert v_1\rangle &=\frac{1}{\sqrt{2}}(\vert \base{1}\rangle +\vert \base{2}\rangle) \\
\vert v_2\rangle &=\frac{1}{\sqrt{2}}(e^{i(\theta'+\theta)}\vert \base{1}\rangle +e^{i(\theta'-\theta)}\vert \base{2}\rangle) .
\end{align}

\subsection{Characterization by genuinely incoherent operations}\Label{S4-2}
Under the multiplicity-free condition,
we characterize our conditions for q-memory uselessness for pure states
by using genuinely incoherent operations (GIO)~\cite{deVicente}, one of the major classes of operations considered in the resource theory of coherence~\cite{streltsov}.
The GIO, which we denote as $\cO_{\rm GIO}$, is defined as the set of operations that preserve all incoherent states, $\Gamma_{\rm GIO}(\sigma)=\sigma$ for every incoherent state $\sigma$, where a state is called incoherent when it is diagonal with a given preferred orthogonal basis.
Then, it was shown in \cite[Thm.2]{deVicente} that $\Gamma$ is in $\cO_{\rm GIO}$ if and only if $\Gamma$ can be written as $\Gamma(\rho) = A\odot\rho$ with some positive semidefinite matrix $A$ with $A_{\lambda\lambda}=1\,\forall \lambda$, 
where $(X\odot Y)_{\lambda\eta}:=X_{\lambda\eta}Y_{\lambda\eta}$ is the Hadamard product. 
 
We find that GIO is a useful tool for analyzing the condition (E2) for general $l$
through the Gram matrix of $V=(\ket{v_1},\dots,\ket{v_l})$ defined as $J(V)_{\lambda\eta}\coloneqq \langle v_\lambda\vert v_\eta\rangle$, $\lambda,\eta = 1,\dots, l$, 
or equivalently, $J(V)\coloneqq V^\dagger V$. 
Since the diagonal element of the Gram matrix $J(V)$ is $1$,
the CPTP map $\Gamma_{J(V)}$ defined as
$ \Gamma_{J(V)}(\rho):= J(V)\odot\rho$ is also in $\cO_{\rm GIO}$.
Then, we have the following characterizations.

\begin{theorem}\Label{pro1}
Suppose $\mathcal{H}_A$ satisfies the multiplicity-free condition. 
Let $\ket{\Psi}=\sum_{\lambda\in\hat G'}\sqrt{P_\Lambda(\lambda)}\ket{\psi_\lambda}\ket{v_\lambda}$ be a given resource state shared by Alice and Bob, where $\hat G'$ with $\vert \hat G'\vert =l$ is a subset of 
irreducible representations for which $P_\Lambda(\lambda)\neq 0\ \forall \lambda\in\hat G'$.
The following conditions 
for $l$ vectors $V :=(\ket{v_1},\dots,\ket{v_l})$ in ${\cal H}_B$
are equivalent.
\begin{description}
\item[(F1)]
The state $\ket{\Psi}$ is q-memory useless.
\item[(F2)]
There exist a basis $\{\ket{\base{k}}\}_{k=1}^{d_B}$ of ${\cal H}_B$, a probability distribution $P_K$ on $\{1,\dots,d_B\}$, and real numbers $\{\theta_{\lambda,k}\}_{\lambda,k}$ that satisfy \eqref{eq:unitary condition}.
\item[(F3)]
There exist a basis $\{\ket{\base{k}}\}_{k=1}^{d_B}$ of ${\cal H}_B$, a probability distribution $P_K$ on $\{1,\dots,d_B\}$, and real numbers $\{\theta_{\lambda,k}\}_{\lambda,k}$ that satisfy 
\begin{align}
    J(V)=J(U)
\Label{eq:condition gram equal}
\end{align}
where $l$ vectors 
$U :=(\ket{u_1},\dots,\ket{u_l})$
are defined as
\begin{align}
    \ket{u_\lambda}\coloneqq \sum_{k=1}^{d_B}\sqrt{P_K(k)}e^{i \theta_{\lambda,k}}\vert \base{k}\rangle
\end{align}
for $\lambda=1,\dots, l$.
\item[(F4)]
There exist a probability distribution $P_K$ on $\{1,\dots,d_B\}$ and real numbers $\{\theta_{\lambda,k}\}_{\lambda,k}$ that satisfy 
\begin{align}
    \frac{1}{l}J(V) = \sum_k P_K(k) Z_k \ket{+}\bra{+} Z_k^\dagger,
\Label{eq:probabilistic incoherent gram}
\end{align}
where
$\ket{+}\coloneqq \frac{1}{\sqrt{l}}\sum_{\lambda=1}^l\ket{\lambda}$ is the maximally coherent state and 
\bal
Z_k\coloneqq \sum_\lambda e^{i\theta_{\lambda,k}}\ket{\lambda}\bra{\lambda}
\eal
is an (incoherent) diagonal unitary. 
\item[(F5)]
There exist a probability distribution $P_K$ on $\{1,\dots,d_B\}$ and real numbers $\{\theta_{\lambda,k}\}_{\lambda,k}$ such that $\Gamma_{J(V)}\in\cO_{\rm GIO}$ can be written as
\bal
\Gamma_{J(V)}(\rho)=\sum_{k} P_K(k) Z_k \rho Z_k^\dagger.
\eal 
\end{description}
\end{theorem}

For Condition (F4), Reference \cite[Thm.3]{deVicente} found that, for $l=2, 3$, any genuinely incoherent operation can be written as a probabilistic incoherent diagonal unitary, while for every $l\geq 4$, there exists a genuinely incoherent operation that cannot be written as a probabilistic incoherent diagonal unitary. 
In addition, for any $\Gamma\in\cO_{\rm GIO}$, there exist 
$l$ vectors $V :=(\ket{v_1},\dots,\ket{v_l})$ such that 
$ \Gamma= \Gamma_{J(V)}$.
Therefore, we obtain the following result.

\begin{corollary}\Label{cor:memory useless dimension}
For every choice of  $\ket{v_1},\dots,\ket{v_l}$, 
there exist a basis $\{\ket{\base{k}}\}_{k=1}^{d_B}$, a probability distribution $\{P_K(k)\}_{k=1}^{d_B}$, 
and real numbers $\{\theta_{\lambda,k}\}_{\lambda,k}$ that satisfy \eqref{eq:unitary condition} for $l=2,\, 3$.
On the other hand, for every $l$ with $l\geq 4$, there exist $\ket{v_1},\dots,\ket{v_l}$ for which there is no choice of $\{\ket{\base{k}}\}_{k=1}^{d_B}$, $\{P_K(k)\}_{k=1}^{d_B}$, and $\{\theta_{\lambda,k}\}_{\lambda,k}$ that satisfies \eqref{eq:unitary condition}. 
In other words,
Bob's quantum memory is not useful to achieve the capacity
$C_{\rm c}(\ket{\Phi})$ for $l=2,3$.
However, when $l \ge 4$, there exists a state 
$\ket{\Phi}$ such that
Bob's quantum memory is useful to achieve the capacity
$C_{\rm c}(\ket{\Phi})$.
\end{corollary}

Therefore, when 
the preshared state $\rho_{AB}$
is a pure state $\vert \Psi\rangle\langle \Psi\vert $ as in Theorem~\ref{pro1},
the multiplicity-free condition holds,
and $\vert \hat{G}'\vert$ is 2 or 3,
the quantum memory ${\cal H}_B$ can be replaced by a classical memory.
However, 
when $\vert\hat{G}'\vert>3$ there is a pure state $\vert \Psi\rangle\langle \Psi\vert $ even with 
the multiplicity-free condition such that 
the quantum memory ${\cal H}_B$ enhances
the communication speed.

\begin{proofof}{Theorem~\ref{pro1}} 
Corollary \ref{Cor1} shows the equivalence between 
Conditions (F1) and (F2). 
Since the application of a unitary operator to a basis
does not change the Gram matrix,
Conditions (F2) and (F3) are equivalent.

Next, we show the equivalence between (F3) and (F4).
Since the Gram matrix $J(U)$ can be written by
\begin{align}
    J(U)_{\lambda\eta}=\sum_{k=1}^{d_B} P_K(k) e^{i\theta_{\lambda,k}-i\theta_{\eta,k}},
    \Label{eq:gram matrix target}
\end{align}
it is written as
\begin{align}
 J(U) = l\sum_{k=1}^{d_B} P_K(k) \ket{w_k}\bra{w_k},
\end{align}
where 
we define $\ket{w_k}=\frac{1}{\sqrt{l}}\sum_{\lambda=1}^{l} e^{i\theta_{\lambda,k}}\ket{\lambda}$
for every $k$.
We further note that $\ket{w_k}$ can be written by $\ket{w_k}=Z_k \ket{+}$, 
This fact implies
\begin{align}
    \frac{1}{l} J(U) = \sum_{k=1}^{d_B} P_K(k) Z_k \ket{+}\bra{+} Z_k^\dagger,
\end{align}
which shows the equivalence between (F3) and (F4).

Next, we show the equivalence between (F4) and (F5).
We assume Condition (F5). 
Then, we have
\begin{align}
 &\frac{1}{l} J(V) =J(V)\odot \ket{+}\bra{+}\nonumber \\
 =& \Gamma_{J(V)}(\ket{+}\bra{+})=
  \sum_{k=1}^{d_B} P_K(k) Z_k \ket{+}\bra{+} Z_k^\dagger,
\end{align}
which implies Condition (F4).

We assume Condition (F4). 
Then, the GIO $\Gamma_{J(V)}$ satisfies
\begin{align}
& \Gamma_{J(V)}(\rho)=J(V) \odot \rho
=\sum_k P_K(k) Z_k \ket{+}\bra{+} Z_k^\dagger \odot \rho \nonumber \\
=&\sum_k P_K(k) \Big(Z_k \ket{+}\bra{+} Z_k^\dagger \odot \rho\Big)
=\sum_k P_K(k) Z_k \rho Z_k^\dagger ,
\end{align}
which implies Condition (F5).
\end{proofof}

\subsection{Construction of q-memory useful pure state}\Label{S4-3}
We provide a systematic method to construct a q-memory useful pure state by using the results of 
GIO.
Since an arbitrary $\frac{1}{l} J({\bf v})$ can be obtained by applying some genuinely incoherent operation to $\ket{+}_l$, 
our problem is reduced to finding whether such an operation in $\cO_{\rm GIO}$ can be implemented by a probabilistic incoherent unitary. 
The following result, which gives the extremality condition for $\cO_{\rm GIO}$, is useful to present examples for which quantum memory is useful for all Bob's bases.

\begin{lemma}[\protect{\cite[Proof of Theorem~21]{deVicente}}]\Label{lem:GIO extremal condition}
 Let $\Gamma\in\cO_{\rm GIO}$ be represented by a set $\{K_t\}_{t=1}^{r'}$ of Kraus operators with $r'>1$. Then, $\Gamma$ 
 cannot be realized by a probabilistic incoherent unitary if 
 ${r'}^2$ operators $\{K^\dagger_t K_s\}_{t,s=1}^{r'}$ are linearly independent.
\end{lemma}

Kraus operators of a channel $\Gamma_{J(V)}\in\cO_{\rm GIO}$
are given as follows.
\begin{lemma}\Label{lem:GIO Kraus2}
Given $l$ vectors $V :=(\ket{v_1},\dots,\ket{v_l})$ in ${\cal H}_B$
with the form $\ket{v_j}=\sum_{k=1}^{d_B} v_{j,k}\ket{\base{k}}$,
the operators $K_k={\rm diag}(\ket{\tilde a_k})$ for vectors 
$\ket{\tilde a_k}=\sum_{j=1}^lv_{j,k} \ket{j}$ with $k=1, \ldots, d_B$ 
form
Kraus operators $\{K_t\}_{t=1}^{d_B}$ of a channel $\Gamma_{J(V)}\in\cO_{\rm GIO}$. 
\end{lemma}
\begin{proof}
Since the relation
\bal
l\sum_{t=1}^{d_B} K_t\ketbra{+}{+}K_t^\dagger = 
\sum_{t=1}^{d_B} \ketbra{\tilde a_t}{\tilde a_t} =
J(V)
\eal
holds, we have
\bal
\sum_{t=1}^{d_B} K_t \rho K_t^\dagger 
= \frac{1}{l}J(V) \odot\rho = \Gamma_{J(V)}(\rho).
\eal
\end{proof}

The above lemma shows that 
Kraus operators of any genuinely incoherent operation $\Gamma_{J(V)}\in\cO_{\rm GIO}$ are given as
diagonal operators.
More generally, the following lemma is known.
\begin{lemma}[\protect{\cite[Theorem~2]{deVicente}}]\Label{lem:GIO Kraus}
    A channel is GIO if and only if all 
    Kraus operators are diagonal.
\end{lemma}

Combining Lemmas~\ref{lem:GIO extremal condition} and \ref{lem:GIO Kraus2} with 
the statement (F5) implies the following. 
When 
the $d_B^2$ vectors $\{ ( \overline{v_{j,t}}  v_{j,s} )_{j=1}^l \}_{t,s=1}^{d_B}$ are linearly independent,
the genuinely incoherent operation $\Gamma_{J(V)}\in\cO_{\rm GIO}$
cannot be realized by a probabilistic incoherent unitary.
Then, a state $\ket{\Psi}$ constructed by the set $\{\ket{v_j}\}_{j=1}^l$ serves as a resource state, for which quantum memory is useful for all Bob's bases.
To satisfy this condition, ${d_B}^2$ needs to be smaller than $l$.

A protocol to find such a resource state is as follows. 
\begin{description}
    \item[1)] Given $l \leq \vert \hat G \vert$, we choose ${d_B}$ such that $1<{d_B}^2\leq l$.
    \item[2)] Pick $l$ states $\ket{v_j}=\sum_{k=1}^{d_B} v_{j,k}\ket{\base{k}}\in\mathbb{C}^{d_B}$ for $j=1,\dots,l$. 
    \item[3)] \Label{item:linear independence} 
    Check whether 
    the ${d_B}^2$ vectors $\{ (  \overline{v_{j,t}}  v_{j,s})_{j=1}^l \}_{t,s=1}^{d_B}$ are linearly independent. If not, go back to the previous step and try another set of $\{\ket{v_j}\}_j$.
    \item[4)] Choose an arbitrary subset of irreps $\hat G'\subset \hat G$ such that $\vert \hat G' \vert =l$. Choose an arbitrary state $\ket{\psi_\lambda}\in\mathcal{H}_\lambda$ for each $\lambda\in\hat G'$ and an arbitrary probability distribution $P_\Lambda(\lambda)$ over $\lambda\in\hat G'$ with $P_\Lambda(\lambda)>0$\ $\forall\lambda\in\hat G'$. 
    \item[5)]  Define $\ket{\Psi}
    =\sum_{j=1}^l\sqrt{P_\Lambda(\lambda_j)}\ket{\psi_{\lambda_j}}\otimes\ket{v_j}
    =\sum_{j=1}^l\sum_{k=1}^{d_B}\sqrt{P_\Lambda(\lambda_j)} v_{j,k} \ket{\psi_{\lambda_j}}\otimes\ket{\base{k}}
    $, where $\lambda_j$ refers to the $j$\,th irrep in $\hat G'$. 
\end{description}

The final state $\ket{\Psi}$ is a state for which quantum memory is useful for all Bob's bases.
For $1<{d_B}^2\leq l$, a randomly chosen $\{\ket{v_j}\}_{j=1}^l$ usually satisfies the linear independence condition in Step~3.  
As an analytical example, we find an instance for 
$\{\ket{v_j}\}_{j=1}^l$
for $l\geq 4$ and ${d_B}=2$ discussed in Ref.~\cite{deVicente}, which is defined as
\begin{align}
v_{j,1}
&:= \left\{
\begin{array}{ll}
\frac{1}{j} & \hbox{ when } j\leq 4\\
1 &  \hbox{ when } j> 4
\end{array}
\right.
\\
v_{j,2}
&:= \left\{
\begin{array}{ll}
i^j\sqrt{1-\frac{1}{j^2}} & \hbox{ when } j\leq 4\\
0 &  \hbox{ when } j> 4.
\end{array}
\right.
\end{align}
One can then check that four vectors 
$\{ ( \overline{v_{j,t}}  v_{j,s} )_{j=1}^l \}_{t,s=1}^2$ are linearly independent.
We then prepare 
$\{\ket{v_j}\}_{j=1}^l$ as
$v_{j,1} \ket{\base{1}}+v_{j,2} \ket{\base{2}}$.

The states
constructed by the above procedure Steps~1--5 provide situations where Bob wants to hold local quantum memory and make a collective measurement on his side before communication.
Notice that, as shown in Ref.~\cite{hayashiWang}, 
the capacity $C_{\rm c}(\rho_{AB})$ for the pure state case 
can be achieved by applying collective measurement across $n$ Bob's local systems of the preshared state.


\section{Case with maximally entangled state}\Label{S5B}
\subsection{Commutative group} 
\Label{S5B-1}
Similarly to Section \ref{S3-3},
we assume that the group $G$ is commutative.
For simplicity, we assume the multiplicity-free condition to the space 
${\cal H}_A$.
We use the same notations as the ones used in Section \ref{S3-3}.
We assume that
the preshared entangled state $\vert \Psi \rangle$ is 
a maximally entangled state $\vert \Phi \rangle$.
Thus, due to the relations \eqref{VG1} and \eqref{VG2},
the condition \eqref{eq:unitary condition} in Corollary \ref{Cor1}
holds.
In this case,
the condition \eqref{eq:unitary condition} in Corollary \ref{Cor1}
shows that
the maximally entangled state $\vert \Phi \rangle$
is q-memory useless
in a simpler way than Condition (A2) does.
This is also a special case of the examples presented in Section \ref{S3-2}.

Next, we proceed to the case when 
the multiplicity-free condition does not hold.
Using an orthogonal basis $\{\vert e_{j,\lambda}\rangle\}_{j=1}^{l(\lambda)}$ where $l(\lambda)\coloneqq \dim\mathcal{M}_\lambda$ and an orthogonal basis $\{\vert v_{j}\rangle\}_{j=1}^{r}$ with
$r=\sum_{\lambda=1}^l l(\lambda)$,
 the maximally entangled state $\vert \Phi \rangle$ is written as
$\frac{1}{\sqrt{r}} 
\sum_{\lambda=1}^l 
\vert\psi_\lambda\rangle 
\sum_{j=1}^{l(\lambda)} 
\vert e_{j,\lambda}\rangle
\vert v_{j+\sum_{\lambda'=1}^\lambda l(\lambda') }\rangle$.
We choose $\vert e_{B,k}\rangle$ as
\begin{align}
\vert e_{B,k}\rangle
:= \frac{1}{\sqrt{r}}\sum_{\lambda=1}^{r} 
e^{2\pi k j i/r }
\vert v_j \rangle.
\end{align}
Then, we have
\begin{align}
\vert \Phi \rangle
=&\frac{1}{\sqrt{r}} 
\sum_{\lambda=1}^l 
\vert\psi_\lambda\rangle 
\sum_{j=1}^{l(\lambda)} 
\vert e_{j,\lambda}\rangle
\frac{1}{\sqrt{r}}\sum_{k=1}^{r} e^{-2\pi k (j+\sum_{\lambda'=1}^\lambda l(\lambda')) i/r }
\vert e_{B,k}\rangle \nonumber \\
=&
\sum_{k=1}^{r} 
\frac{1}{\sqrt{r}} 
\sum_{\lambda=1}^l 
\frac{\sqrt{l(\lambda)}}{\sqrt{r}} 
e^{-2\pi k \sum_{\lambda'=1}^\lambda l(\lambda') i/r }
\vert\psi_\lambda\rangle 
\Big(\frac{1}{\sqrt{l(\lambda)}}
\sum_{j=1}^{l(\lambda)} 
e^{-2\pi k j i/r }
\vert e_{j,\lambda}\rangle \Big)
\vert e_{B,k}\rangle .
\end{align}
Hence, by choosing 
$\vert\psi_\lambda\rangle \Big(\frac{1}{\sqrt{l(\lambda)}}
\sum_{j=1}^{l(\lambda)} 
e^{-2\pi k j i/r }
\vert e_{j,\lambda}\rangle \Big)$
as $ \vert \psi_{\lambda,k} \rangle$,  
Condition (C1) holds in Theorem \ref{TH1P}.
That is, the quantum memory is not needed even in this case.

In particular, when the set $\hat{G}'$ is the whole set $\hat{G}$ of irreducible representations of $G$,
the entropy of
$\sum_{g \in G}\frac{1}{\vert G \vert} 
U_g  \vert \Phi\rangle \langle \Phi\vert U_g^\dagger$
is $\log \vert G \vert$.
This means that
the action of $G$ can be distinguished perfectly
even when Bob measures his memory in advance.
In other words,
the states $\{    U_g  \vert \Phi\rangle \}_{g \in G}$
can be distinguished perfectly
by one-way local measurements.
For example,
given Weyl-Heisenberg representation of 
$\mathbb{Z}_d \times \mathbb{Z}_d$,
the action of the commutative subgroup of 
$\mathbb{Z}_d \times \mathbb{Z}_d$
can be distinguished perfectly
by one-way local measurements.
As another example, 
given Weyl-Heisenberg representation of 
$\mathbb{F}_p^n \times \mathbb{F}_p^n$,
the action of the commutative subgroup of 
$\mathbb{F}_p^n \times \mathbb{F}_p^n$
can be distinguished perfectly by one-way local measurements.
In fact, this kind of state discrimination has 
been discussed in the context of local discrimination of generalized Bell states
in 
\cite{PhysRevLett.85.4972,
PhysRevLett.87.277902,
10.1063/1.1914731,
PhysRevA.74.032108,
PhysRevLett.92.177905,
PhysRevA.103.052429,
PhysRevA.105.032455,Yang_2022}.

While the classification of a commutative subgroup
in $\mathbb{Z}_d \times \mathbb{Z}_d$ is not easy \cite{PhysRevA.103.052429,
PhysRevA.105.032455,Yang_2022},
a commutative subgroup in $\mathbb{F}_p^n \times \mathbb{F}_p^n$
can be easily classified as follows.
Commutativity can be characterized by 
the orthogonality under the symplectic inner product.
The set of $m$-dimensional orthogonal subspace
can be considered as 
an orthogonal version of 
Grassmannian over finite fields, whose classification 
is well known \cite[Proposition 1.7.2]{Stanley}.
For the identification of 
$m$-dimensional orthogonal subspace, 
it is sufficient to choose $m$ 
independent commutative vectors in 
$\mathbb{F}_p^n \times \mathbb{F}_p^n$.
The choice of the first non-zero vector has $p^{2n}-1$ cases.
The choice of the second non-zero vector has $p^{2(n-1)}-1$ cases
due to the commutativity with the first vector.
The commutativity with the first vector is equivalent to 
the orthogonal property to the first vector
for the symplectic inner product.
Therefore, 
in the case of $m$-dimensional subgroups, 
we have 
$(p^{2n}-1)(p^{2(n-1)}-1)\cdots (p^{2(n-m+1)}-1)$ cases.
Since 
$(p^m-1)(p^{m-1}-1)\cdots (p-1)$ choices of 
$m$ independent commutative vectors 
correspond to the same $m$-dimensional subgroup, 
we have 
$\frac{(p^{2n}-1)(p^{2(n-1)}-1)\cdots (p^{2(n-m+1)}-1)}{(p^m-1)(p^{m-1}-1)\cdots (p-1)}$
commutative 
$m$-dimensional subgroup in total.
When $m=n$, this number is simplified to
$(p^{n}+1)(p^{n-1}+1)\cdots (p+1)$.

\subsection{General group}\Label{S5B-2}
Next, we observe how the above fact can be generalized to a general group
$G$.
When $G$ is not a commutative group,
the irreducible representation is not one-dimensional.
Hence, this case cannot be considered as a special case of the examples presented in Section \ref{S3-2}.

We have the following lemma
for a general group with respect to
the relation between Condition (C1) and
a maximally entangled state $\vert\Phi\rangle$
when the set $\hat{G}'$ of irreducible representations appearing in 
${\cal H}_A$ is given as $\{1, \ldots, l\}$.

\begin{lemma}
A maximally entangled state 
$\vert\Phi\rangle$ satisfies
Condition (C1) in Theorem \ref{TH1P}
if and only if 
the relation $\dim {\cal M}_{\lambda} \ge 
\dim {\cal H}_{\lambda}$ holds for $\lambda \in \hat{G}'$.
\end{lemma}

This lemma guarantees the following.
When 
the preshared entangled state $\vert\Psi\rangle$ is a maximally entangled state 
$\vert\Phi\rangle$
and 
the relation $\dim {\cal M}_{\lambda} \ge 
\dim {\cal H}_{\lambda}$ holds for $\lambda \in \hat{G}'$,
the states $\{    U_g  \vert \Phi\rangle \}_{g \in G}$
can be distinguished perfectly
by one-way local measurements.

\begin{proof}
The ``only if'' part.
We assume that a maximally entangled state 
$\vert\Phi\rangle$ satisfies the condition \eqref{AC8}.
Then, the condition \eqref{AC8} implies 
$\tr_{{\cal M}_{\lambda}} \vert \psi_{\lambda,k}\rangle\langle \psi_{\lambda,k}\vert  $ is 
$\tr_{{\cal M}_{\lambda}}  \rho_{A}$.
Hence, 
$\tr_{{\cal M}_{\lambda}} \vert \psi_{\lambda,k}\rangle\langle \psi_{\lambda,k}\vert  $ is the completely mixed state.
Thus,
$\dim {\cal M}_{\lambda} \ge 
\dim {\cal H}_{\lambda}$.

The ``if'' part.
We assume that
$\dim {\cal M}_{\lambda} \ge 
\dim {\cal H}_{\lambda}$ for $\lambda \in \hat{G}'$.
When two entangled states have the same reduced density matrix on ${\cal H}_A$,
these two entangled states can be converted to each other via a
local unitary on ${\cal H}_B$.
It is sufficient to show the existence of the choice of $\{\vert \base{k}\rangle \}$
and $\{\vert \psi_{\lambda,k}\rangle \}$
to satisfy Condition (C1).

We denote the dimension of ${\cal M}_\lambda$ by $d_\lambda'$.
We choose 
an orthogonal basis $\{ \vert e_{\lambda,A,j}\rangle\}_{j=0}^{d_\lambda-1}$
of ${\cal H}_\lambda$,
and 
an orthogonal basis $\{ \vert e_{\lambda,A,j'}'\rangle\}_{j'=0}^{d_\lambda'-1}$
of ${\cal M}_\lambda$.
We define the operators
$Z:=
\sum_{j=0}^{d_\lambda-1}
e^{2\pi ij/d_\lambda}
\vert e_{\lambda,A,j}\rangle \langle e_{\lambda,A,j}\vert$
on ${\cal H}_\lambda$,
and
$X:=
\sum_{j'=0}^{d_\lambda'-1}
\vert e_{\lambda,A,j'+1}'\rangle \langle e_{\lambda,A,j'}'\vert$
on ${\cal M}_\lambda$, where
modulo $d_\lambda'$ is applied in the definition of $X$.
For $(k_\lambda,k_\lambda') \in \{0, 1,\ldots,d_{\lambda}-1\}
\times \{0, 1,\ldots,d_{\lambda}'-1\}$,
we define the vectors $\vert \psi_{\lambda,k,k'}\rangle
:=Z^k \otimes X^{k'}
\frac{1}{\sqrt{d_\lambda}}\sum_{j=0}^{d_\lambda-1}
\vert e_{\lambda,A,j}\rangle
\vert e_{\lambda,A,j}'\rangle$ in 
${\cal H}_\lambda \otimes {\cal M}_\lambda$, which forms
an orthogonal basis of ${\cal H}_\lambda \otimes {\cal M}_\lambda$.

For $(k_\lambda,k_\lambda') \in \{0, 1,\ldots,d_{\lambda}-1\}
\times \{0, 1,\ldots,d_{\lambda}'-1\} $ with $\lambda \in 
\hat{G}'=\{1,2, \ldots, l\}$,
we define the vector
$\vert e_{B, k_1,k_1', \ldots, k_l,k_l'},\lambda\rangle$
in ${\cal H}_B$.
Then, we define the vector 
\begin{align}
& \sum_{k_1=0}^{d_1-1}
\sum_{k_1'=0}^{d_1'-1}
\cdots
\sum_{k_l=0}^{d_l-1}
\sum_{k_l'=0}^{d_l'-1}
\sum_{\lambda'=1}^l
\sqrt{\frac{1}{
d_1d_1' \cdots d_l d_l' l}} \nonumber \\
&\quad \cdot
\sum_{\lambda=1}^l
\sqrt{
\frac{d_{\lambda} d_{\lambda}'}{\sum_{\lambda''=1}^{l} d_{\lambda''} d_{\lambda''}'}}
e^{2\pi \lambda \lambda' i/l}
\vert \psi_{\lambda,k_\lambda,k_\lambda'}\rangle
\vert e_{B, k_1,k_1', \ldots, k_l,k_l'},\lambda'\rangle,
\Label{MNG}
\end{align}
which satisfies Condition (C1).

The reduced density matrix of 
\begin{align*}
\sum_{\lambda'=1}^l
\sqrt{\frac{1}{l}}
\sum_{\lambda=1}^l
\sqrt{
\frac{d_{\lambda} d_{\lambda}'}{\sum_{\lambda''=1}^{l} d_{\lambda''} d_{\lambda''}'}}
e^{2\pi \lambda \lambda' i/l}
\vert \psi_{\lambda,k_\lambda,k_\lambda'}\rangle
\vert e_{B, k_1,k_1', \ldots, k_l,k_l'},\lambda'\rangle.
\end{align*}
on ${\cal H}_A$ is
\begin{align}
\sum_{\lambda=1}^l
\frac{d_{\lambda} d_{\lambda}'}{\sum_{\lambda''=1}^{l} d_{\lambda''} d_{\lambda''}'}
\vert \psi_{\lambda,k_\lambda,k_\lambda'}\rangle
\langle \psi_{\lambda,k_\lambda,k_\lambda'}\vert.
\end{align}
Therefore, 
the reduced density matrix of the vector \eqref{MNG} 
on ${\cal H}_A$ is
\begin{align}
\sum_{\lambda=1}^l
\sum_{k_\lambda=0}^{d_\lambda-1}
\sum_{k_\lambda'=0}^{d_\lambda'-1}
\frac{1}{\sum_{\lambda''=1}^{l} d_{\lambda''} d_{\lambda''}'}
\vert \psi_{\lambda,k_\lambda,k_\lambda'}\rangle
\langle \psi_{\lambda,k_\lambda,k_\lambda'}\vert,
\end{align}
which is the completely mixed state on ${\cal H}_A$.
Therefore, a maximally entangled state 
$\vert\Phi\rangle$ 
has the form to satisfy Condition (C1).
Hence, we obtain the ``if'' part.
\end{proof}

\section{Implications to quantum illumination}\Label{S5}
As one of the settings of physical significance, 
quantum illumination investigates
whether a low-reflective target object is present by shooting light and collecting the signal.
The problem of target detection can be reduced to 
the channel discrimination problem to identify whether the given channel is 
a replacement channel that prepares the thermal state or
a thermal attenuator channel.  
In this scenario, it has been shown that entanglement in the input state can enhance the discrimination performance in symmetric and asymmetric channel discrimination settings~\cite{Lloyd2008enhanced,Tan2008quantum,dePalma}. 

Here, we consider a variant of quantum illumination, 
where instead of detecting the thermal noise, we  
are to detect the existence of a certain noise channel\footnote{Our setting can also be considered as a variant of quantum reading~\cite{Pirandola2011quantum} in the sense that the reflectivity of the target is not small.}. 
Namely, we consider the asymmetric channel discrimination task
where the null hypothesis is error-free, i.e., 
$\Gamma_1=\id$, and the alternative hypothesis is to have a group twirling channel as noise,  i.e., $\Gamma_2=\mathcal{G}$. 
One can then ask whether entanglement in a given input state $\psi$ helps the hypothesis testing, i.e.,  whether one would like to hold a local quantum memory to utilize the entanglement in $\psi$.

This problem is formulated as follows.
Let $\eta_{j}:=\Gamma_{j}\otimes\id(\psi)$ be the output states with an input state $\psi$
for $j=1,2$. 
In general, one is allowed to make a collective POVM measurement $\{E_n,I-E_n\}$ on multiple copies $\eta_{j}^{\otimes n}$ of output states. 
We define the type-I error by $\alpha_n(E):=\tr((I-E)\,\eta_1^{\otimes n})$ and type-II error by $\beta_n(E):=\tr(E\,\eta_2^{\otimes n})$.
In the asymmetric channel discrimination task, we aim to minimize the type-II error under the condition that type-I error is upper bounded by a constant $\epsilon$.
The Stein's lemma~\cite{hiai1991proper,ogawa2000strong} tells that for all $0<\epsilon<1$, the error exponent of the type-II error is characterized by the relative entropy between two output states, i.e.,
\begin{align}
 &-\lim_{n\to\infty}\frac{1}{n}\log\inf_{0\leq E_n\leq I} \left\{\beta_n(E_n) : \alpha_n(E_n)\leq 1-\epsilon\right\} \nonumber \\
 =& D(\Gamma_1\otimes\id(\psi)\|\Gamma_2\otimes\id(\psi)).
 \Label{eq:stein}
\end{align}

We can address this question for the case of asymmetric channel discrimination by employing the framework established above. Since Stein's lemma assures that the performance of the hypothesis testing is characterized by the relative entropy as in \eqref{eq:stein}, entanglement in $\psi$ is helpful (one would like to retain local quantum memory) if and only if the equality in \eqref{XL1} with $\rho_{AB}=\psi$ holds.  
We can also carry over the characterization of the resource state in Theorem~\ref{pro1} and Corollary~\ref{cor:memory useless dimension} for the local memory to be useless.

As a typical example, we choose $G=\mathbb{Z}_d$ with a representation 
$U_g=\sum_{t=0}^{d-1}\ketbra{t+g\ {\rm mod}\ d}{t}$. 
This representation can be decomposed into $U_g = \bigoplus_{\lambda=0}^{d-1} U_{\lambda, g}$ where $U_{\lambda, g}:= e^{2\pi i g\lambda/d}\dm{+_\lambda}$ and $\ket{+_\lambda}:= \frac{1}{\sqrt{d}}\sum_{j=0}^{d-1} e^{2\pi ij \lambda/d}\ket{j}$.
Then, Theorem~\ref{pro1}, as well as a consequence from Sec.~\ref{S5B-1}, implies that the maximally entangled input $\ket{\psi}=\frac{1}{\sqrt{d}}\sum_{t=0}^{d-1}\ket{tt}=\frac{1}{\sqrt{d}}\sum_{\lambda=0}^{d-1}\ket{+_\lambda}\ket{+_{d-\lambda-1}}$ does not require the local quantum memory to achieve the optimal performance for the noise detection task. 
Indeed, the Gram matrix $J({\bf v})$ in Theorem~\ref{pro1} is the $d\times d$ identity matrix, as $\{\ket{+_\lambda}\}_{\lambda}$ constructs an orthonormal basis.  
Therefore, by taking $P_K(k) = 1/d$ and $\theta_{\lambda,k}=2\pi \lambda k/d$ in (F4) of Theorem~\ref{pro1}, we can check that \eqref{eq:probabilistic incoherent gram} is indeed satisfied. 
This makes a stark contrast to the case of the original quantum illumination, in which the two-mode squeezed state, which is the infinite-dimensional correspondence of the maximally entangled state, serves as the optimal input with the help of quantum memory~\cite{dePalma}.

\section{Discussion}\label{S7}
We have introduced the new concept of q-memory uselessness and usefulness of a preshared bipartite state in the dense-coding task based on a given group representation.
We have derived various conditions for q-memory uselessness for general mixed and pure states.
Using our general conditions, we have presented an example of a mixed entangled state that is q-memory useless.
In addition, we have revealed a notable relation between 
q-memory uselessness of pure states and the resource theory of coherence.
This relation is useful when the state is pure, the group is abelian, and the representation is multiplicity-free.
Under this condition, we have shown that 
any pure state is q-memory useless
when 
the given group representation consists of at most three irreducible representations.
We also have presented a systematic way to construct a pure q-memory useful state 
when the group representation consists of more than three irreducible representations.
Finally, we have discussed the relations between our framework and quantum illumination. We have considered a variant of quantum illumination that aims to distinguish the noiseless and the group twirling channels and employed our conditions to characterize the usefulness of quantum memory in enhancing the performance.

An interesting future direction is to find whether there exists a separable q-memory useful state.
The class of q-memory useless states represented by the form in \eqref{ZCXOT} is known to have zero quantum discord~\cite{modi2012classical}.  
This implies that q-memory useful separable states, if any, should have nonzero quantum discord.
Establishing the quantitative relation between quantum discord and q-memory usefulness may provide further insights into the role of quantum correlation in dense coding.

\backmatter

\bmhead{Acknowledgments}
M.H. is very grateful to Dr. Jiawei Wu for helpful discussions 
for the usefulness of quantum memory in private dense coding. 
We also thank an anonymous reviewer for suggesting a simpler proof of Proposition~\ref{LLA}.
R.T. acknowledges the support of the Lee Kuan Yew Postdoctoral Fellowship at Nanyang Technological University Singapore. 
M.H. is supported in part by
the National Natural Science Foundation of China (Grant
No. 62171212).

\bmhead{Data availability}
Data sharing is not applicable to this article as no datasets were generated
or analyzed during the current study.

\bmhead{Conflict of interest}
There are no competing interests.


\begin{appendices}

\section{Derivation from (C1) to (A2)}
We demonstrate how to derive from (C1) to (A2) in the pure state case.
That is, we check that a pure state $\rho_{AB}=\vert  \Psi\rangle\langle\Psi\vert  $ that satisfies \eqref{eq:equality condition 1 pure} and \eqref{AC8} indeed satisfies \eqref{eq:equality condition general} as follows.

Let $\ket{\psi_{\lambda,k}}=\sum_j \sqrt{c_{\lambda,k,j}}\ket{r_{\lambda, k, j}}_{\cal H_\lambda}\ket{m_{\lambda, k, j}}_{\mathcal{M}_\lambda}$ be a Schmidt decomposition of $\ket{\psi_{\lambda,k}}$, where $\left\{\ket{r_{\lambda, k ,i}}\right\}_i$ and $\left\{\ket{m_{\lambda, k, i}}\right\}_i$ are orthonormal bases on ${\cal H_\lambda}$ and $\mathcal{M}_\lambda$.
Then, the condition \eqref{AC8} ensures that  
\begin{align}
    \sum_j c_{\lambda,k,j} \dm{r_{\lambda, k, j}}= \sum_j c_{\lambda,k',j} \dm{r_{\lambda,k',j}},\ \forall k,k'.
\label{eq:condition reduced state}
\end{align}
Substituting this into the form  \eqref{eq:equality condition 1 pure} yields
\begin{align}
    \ket{\Psi} =& \sum_\lambda \sqrt{P_\Lambda(\lambda)} \sum_k \sqrt{P_{K}(k)}e^{i\theta_{\lambda,k}}\nonumber \\
&  \cdot  \sum_j \sqrt{c_{\lambda,k,j}}\ket{r_{\lambda, k, j}}_{\cal H_\lambda}
    \ket{m_{\lambda, k, j}}_{\mathcal{M}_\lambda}\ket{\base{k}}_B.
    \Label{eq:pure resource Schmidt}
\end{align}
Using this expression, we can write
\begin{align}
 {\cal B}(\rho_{AB}) 
    = \sum_{\lambda\lambda'k}\sqrt{P_\Lambda(\lambda)P_\Lambda(\lambda')} P_{K}(k)e^{i\theta_{\lambda,k}-i\theta_{\lambda',k}}\vert  \psi_{\lambda,k}\rangle\langle\psi_{\lambda',k}\vert  \otimes\vert  \base{k}\rangle\langle \base{k}\vert.
\end{align}

\begin{align}
    {\cal G}\circ{\cal B}(\rho_{AB}) 
    = \sum_{\lambda k}P_\Lambda(\lambda) P_{K}(k)\sum_j c_{\lambda,k,j} \frac{I}{d_{\cal H_\lambda}}\otimes \vert m_{\lambda, k, j}\rangle\langle m_{\lambda, k, j}\vert \otimes\vert \base{k}\rangle\langle \base{k}\vert 
\end{align}

\begin{align}
\sqrt{{\cal G}\circ{\cal B}(\rho_{AB})}^{-1} 
    = \sum_{\lambda k}\sqrt{P_\Lambda(\lambda) P_{K}(k)}^{-1}\sum_{j:c_{\lambda,k,j}\neq 0} \sqrt{\frac{d_{\cal H_\lambda}}{c_{\lambda,k,j}}}I\otimes \vert m_{\lambda, k, j}\rangle\langle m_{\lambda, k, j}\vert \otimes\vert \base{k}\rangle\langle \base{k}\vert .
\end{align}
We then have
\begin{equation}\begin{aligned}
&\sqrt{{\cal G}\circ{\cal B} (\rho_{AB})}^{-1}{\cal B}(\rho_{AB})\sqrt{{\cal G}\circ{\cal B} (\rho_{AB})}^{-1}\\
    &\quad= \sum_{\lambda\lambda'k} e^{i\theta_{\lambda,k}-i\theta_{\lambda',k}}d_{\cal H_\lambda}\vert \tilde \Phi_{\lambda, k}\rangle\langle\tilde\Phi_{\lambda',k}\vert \otimes\vert \base{k}\rangle\langle \base{k}\vert 
    \Label{eq:pure Petz 1}
\end{aligned}\end{equation}
where $\ket{\tilde \Phi_{\lambda,k}}\coloneqq\sum_{j:c_{\lambda,k,j}\neq 0}\ket{r_{\lambda, k, j}}\ket{m_{\lambda, k, j}}$.

We also have 
\begin{align}
    {\cal G}(\rho_{AB}) = \sum_{\lambda}P_\Lambda(\lambda) \frac{I}{d_{\cal H_\lambda}}\otimes\tr_{\mathcal{H}_\lambda}\vert \phi_{\lambda}\rangle\langle\phi_{\lambda}\vert  
\end{align}
and 
\begin{align}
    \sqrt{{\cal G}(\rho_{AB})} = \sum_{\lambda}\sqrt{P_\Lambda(\lambda)} \frac{I}{\sqrt{d_{\cal H_\lambda}}}\otimes\sqrt{\tr_{\mathcal{H}_\lambda}\vert \phi_{\lambda}\rangle\langle\phi_{\lambda}\vert}  
\end{align}
where $\ket{\phi_{\lambda}} \coloneqq \sum_k \sqrt{P_{K}(k)}e^{i\theta_{\lambda,k}}\sum_j \sqrt{c_{\lambda,k,j}}\ket{r_{\lambda, k, j}}_{\cal H_\lambda}\ket{m_{\lambda, k, j}}_{\mathcal{M}_\lambda}\ket{\base{k}}_B$.

The fact that 
\bal
\sqrt{\tr_{\mathcal{H}_\lambda}\vert \phi_{\lambda}\rangle\langle\phi_{\lambda}\vert} &= \sum_{kk'} \sqrt{P_{K}(k)P_K(k')}e^{i\theta_{\lambda,k}-i\theta_{\lambda,k'}}\sum_{jj'} \left(c_{\lambda,k,j} c_{\lambda,k',j'}\right)^{1/4}\\
&\qquad \braket{r_{\lambda,k',j'}}{r_{\lambda,k,j}}\ketbra{m_{\lambda, k, j}}{m_{\lambda,k',j'}}\otimes\ketbra{\base{k}}{\base{k'}}
\label{eq:square root}
\eal
can be seen as follows. 
The square of the right-hand side becomes
\begin{equation}\begin{aligned}
&\sum_{kk'k''} \sqrt{P_{K}(k)P_K(k'')}P_K(k')e^{i\theta_{\lambda,k}-i\theta_{\lambda,k''}}\sum_{jj'j''} \left(c_{\lambda,k,j} c_{\lambda,k'',j''}\right)^{1/4}\sqrt{c_{\lambda,k',j'}}\\
&\qquad \braket{r_{\lambda,k',j'}}{r_{\lambda,k,j}}\braket{r_{\lambda,k'',j''}}{r_{\lambda,k',j'}}\ketbra{m_{\lambda, k, j}}{m_{\lambda,k'',j''}}\otimes\ketbra{\base{k}}{\base{k''}}\\
&=\sum_{kk'k''} \sqrt{P_{K}(k)P_K(k'')}P_K(k')e^{i\theta_{\lambda,k}-i\theta_{\lambda,k''}}\sum_{jj''} \left(c_{\lambda,k,j} c_{\lambda,k'',j''}\right)^{1/4}\\
&\qquad \bra{r_{\lambda,k'',j''}}\left(\sum_{j'}\sqrt{c_{\lambda,k',j'}}\ket{r_{\lambda,k',j'}}\bra{r_{\lambda,k',j'}}\right)\ket{r_{\lambda,k,j}}\ketbra{m_{\lambda, k, j}}{m_{\lambda,k'',j''}}\otimes\ketbra{\base{k}}{\base{k''}}.
\label{eq:square of square root first}
\end{aligned}\end{equation}
Taking the matrix power in both sides of \eqref{eq:condition reduced state} results in
\bal
 \sum_j \left(c_{\lambda,k,j}\right)^a \dm{r_{\lambda, k, j}}= \sum_j \left(c_{\lambda,k',j}\right)^a \dm{r_{\lambda,k',j}},\ \forall k,k'
\label{eq:power}
\eal
for every real number $a$.
This implies 
\begin{equation}\begin{aligned}
 &\sum_{j'}\sqrt{c_{\lambda,k',j'}} \dm{r_{\lambda,k',j'}} \\
 &= \left(\sum_{j'} \left(c_{\lambda,k',j'}\right)^{1/4} \dm{r_{\lambda,k',j'}}\right)\left(\sum_{j'} \left(c_{\lambda,k',j'}\right)^{1/4} \dm{r_{\lambda,k',j'}}\right)\\
 &=\left(\sum_{j'} \left(c_{\lambda,k'',j'}\right)^{1/4} \dm{r_{\lambda,k'',j'}}\right)\left(\sum_{j'}\left(c_{\lambda,k,j'}\right)^{1/4} \dm{r_{\lambda,k,j'}}\right).
 \label{eq:decomposing square root}
\end{aligned}\end{equation}

Substituting \eqref{eq:decomposing square root} into \eqref{eq:square of square root first} gives 
\begin{equation}\begin{aligned}
&\sum_{kk'k''} \sqrt{P_{K}(k)P_K(k'')}P_K(k')e^{i\theta_{\lambda,k}-i\theta_{\lambda,k''}}\sum_{jj''} \sqrt{c_{\lambda,k,j} c_{\lambda,k'',j''}}\\
&\qquad \braket{r_{\lambda,k'',j''}}{r_{\lambda,k,j}}\ketbra{m_{\lambda, k, j}}{m_{\lambda,k'',j''}}\otimes\ketbra{\base{k}}{\base{k''}}\\
&=\sum_{kk''} \sqrt{P_{K}(k)P_K(k'')}e^{i\theta_{\lambda,k}-i\theta_{\lambda,k''}}\sum_{jj''} \sqrt{c_{\lambda,k,j} c_{\lambda,k'',j''}}\\
&\qquad \braket{r_{\lambda,k'',j''}}{r_{\lambda,k,j}}\ketbra{m_{\lambda, k, j}}{m_{\lambda,k'',j''}}\otimes\ketbra{\base{k}}{\base{k''}}\\
&=\tr_{\mathcal{H}_\lambda}\dm{\phi_\lambda},
\end{aligned}\end{equation}
which confirms \eqref{eq:square root}.

Combing this with \eqref{eq:pure Petz 1} and noting that 
\begin{equation}\begin{aligned}    I\otimes\sqrt{\tr_{\mathcal{H}_\lambda}\dm{\phi_\lambda}}\ket {\tilde\Phi_{\lambda,k}}\ket{\base{k}}&= \sum_{k'} \sqrt{P_{K}(k')P_K(k)}e^{i\theta_{\lambda,k'}-i\theta_{\lambda,k}}\sum_{jj'} \left(c_{\lambda,k,j} c_{\lambda,k',j'}\right)^{1/4}\\
&\qquad \braket{r_{\lambda,k,j}}{r_{\lambda,k',j'}}\ket{r_{\lambda,k,j}}\ket{m_{\lambda,k',j'}}\otimes\ket{\base{k'}},
\end{aligned}\end{equation}
we get
\begin{equation}\begin{aligned}
    &\sqrt{{\cal G}(\rho_{AB})}\sqrt{{\cal G}\circ{\cal B} (\rho_{AB})}^{-1}{\cal B}(\rho_{AB})\sqrt{{\cal G}\circ{\cal B} (\rho_{AB})}^{-1}\sqrt{{\cal G}(\rho_{AB})} \\
    &= \sum_{\lambda\lambda'k}\sqrt{P_\Lambda(\lambda)P_\Lambda(\lambda')} e^{i\theta_{\lambda,k}- i\theta_{\lambda',k}}\sum_{k'k''} \sqrt{P_{K}(k')P_K(k'')}P_K(k)e^{i\theta_{\lambda,k'}-i\theta_{\lambda,k} + i\theta_{\lambda',k}-i\theta_{\lambda',k''}}\\
    &\quad \sum_{jj'}\left(c_{\lambda,k,j} c_{\lambda,k',j'}\right)^{1/4}\braket{r_{\lambda,k,j}}{r_{\lambda,k',j'}}\sum_{lj''}\left(c_{\lambda',k,l} c_{\lambda',k'',j''}\right)^{1/4}\braket{r_{\lambda',k'',j''}}{r_{\lambda',k,l}}\\
    &\qquad\ketbra{r_{\lambda,k,j}}{r_{\lambda',k,l}}\otimes\ketbra{m_{\lambda,k',j'}}{m_{\lambda',k'',j''}}\otimes\ketbra{\base{k'}}{\base{k''}} \\
    &= \sum_{\lambda\lambda'}\sqrt{P_\Lambda(\lambda)P_\Lambda(\lambda')} \sum_{k'k''} \sqrt{P_{K}(k')P_K(k'')}e^{i\theta_{\lambda,k'}-i\theta_{\lambda',k''}}\\
    &\quad \sum_{j'j''}\sqrt{c_{\lambda,k',j'}}\sqrt{c_{\lambda',k'',j''}}\ketbra{r_{\lambda,k',j'}}{r_{\lambda',k'',j''}}\otimes\ketbra{m_{\lambda,k',j'}}{m_{\lambda',k'',j''}}\otimes\ketbra{\base{k'}}{\base{k''}} \\
    &=\vert \Psi\rangle\langle\Psi\vert ,
\end{aligned}\end{equation}
where in the second equality we used \eqref{eq:power}.
This confirms \eqref{eq:equality condition general}.
\end{appendices}

\bibliography{references}

\end{document}